\documentclass[a4paper,12pt]{article}
\usepackage[letter paper, margin=1 in]{geometry}
\usepackage{hyperref}
\usepackage{subcaption}
\usepackage{authblk}
\usepackage[numbers]{natbib}
\usepackage[utf8]{inputenc}
\usepackage{graphicx}
\usepackage{amsmath}
\usepackage[version=4]{mhchem}
\usepackage{siunitx}
\usepackage{longtable,tabularx}
\usepackage{microtype}
\usepackage{subcaption}
\usepackage{booktabs} 
\usepackage{float}
\usepackage[algo2e,ruled,vlined,linesnumbered]{algorithm2e}
\usepackage{amsfonts,amsthm}
\usepackage{bm,bbm}
\usepackage{cleveref}
\usepackage{algorithm}
\usepackage{algorithmic}
\usepackage[algo2e,ruled,vlined,linesnumbered]{algorithm2e}
\usepackage{svg}
\usepackage{tikz}
\usetikzlibrary{positioning,fit,calc,shadows}
\pgfdeclarelayer{base}
\pgfdeclarelayer{bg}
\pgfdeclarelayer{fg}
\pgfsetlayers{base,bg,main,fg}
\usepackage{url}
\usepackage{xcolor}

\newcommand{\x}{\mathbf{x}}

\newtheorem{proposition}{Proposition}

\title{REMAL: Residual Equilibrium Manifold Active Learning for Surrogate-Based Multidisciplinary Design Analysis}
\author{%
    Kail Yuan \thanks{Aerospace Engineering, Penn State, University Park, PA.} 
  \quad Ashwin Renganathan\thanks{Aerospace Engineering and the Institute for Computational and Data Sciences (ICDS), Penn State, University Park, PA. Corresponding author.}
}
\date{}

\begin{document}

\maketitle

\begin{abstract}
Multidisciplinary design analysis of coupled engineering systems requires the computation of equilibrium states in which all disciplinary coupling variables are mutually consistent. Conventional fixed-point iteration resolves this consistency problem separately at each design point, which can become expensive when disciplinary evaluations are costly and many analyses are required in outer-loop tasks such as multidisciplinary design optimization, uncertainty quantification, or digital twin updating. This paper introduces REMAL, a residual manifold surrogate modeling framework for coupled systems. Instead of approximating each discipline independently or directly learning converged coupling variables, the proposed method learns a surrogate model of the joint residual manifold via multitask Gaussian process models. An entropy-based active learning strategy selects additional residual evaluations near uncertain zero-contour regions, and equilibrium states for new design inputs are recovered by solving a nonlinear least squares optimization problem using only the trained surrogate. The method is evaluated on four engineering coupled system benchmarks: a satellite model, an aerostructural model, a finite-element gas-turbine heat-transfer and economics model, and a modified turbine model with added feedback coupling. Across these cases, REMAL consistently demonstrates the cost effectiveness when repeated evaluations of the fixed point across the design space are necessary. Theoretically, we show that, under mild assumptions, REMAL's predictive fixed point error is bounded.
\end{abstract}

\section{Introduction}
\label{sec:intro}
Engineered systems are typically multidisciplinary, involving two or more disciplines interacting nonlinearly. For example, an aircraft can be thought of as a system involving aerodynamics, structures, propulsion, and other subsystems tightly integrated leading to nonlinear interactions and system-level tradeoffs. 
Designing such systems while accounting for the underlying interdisciplinary interactions is widely known as multidisciplinary design optimization (MDO)~\cite{sobieszczanski-sobieski_sensitivity_1990}. 
These systems are characterized by \textit{coupling} between individual disciplines or ``components'', which may be either feed-forward (unidirectional) or feedback (bidirectional) in nature.
Mathematically, coupled systems can be represented as 
\begin{equation}
    \begin{alignedat}{2}
        y_i =& f_i(\mathbf{x}, \{y_j\}_{j\neq i}),~ \forall i, j\in [m] \qquad &\text{(state)} \\
        z =& g(\mathbf{x}, \{y_i^*\}_{i=1}^m) \qquad &\text{(output)}
        \label{eqn:md_system}
    \end{alignedat}
\end{equation}
where $\mathbf{x} \in \mathcal{X} \subset \mathbb{R}^d$ are the control parameters, $y_i \in \mathcal{Y}_i \subset 
\mathbb{R}^o$ are the state or ``coupling'' variables, $f_i:\mathcal{X}\times \mathcal{Y}_{j \neq i} \rightarrow \mathbb{R}^o,~\forall i,j \in[m]$ are the $m$ disciplines, and $z\in \mathbb{R}$ is the output. Typically, the \(f_i\) are physics-based simulation codes such as computational fluid dynamics (CFD) or finite-element method (FEM) solvers that are expensive to evaluate and expose no structure (that is, they are treated as black boxes).
A unique requirement for solving coupled multidisciplinary systems is the need to compute an
\emph{equilibrium state} $y_i^*$ for the system that satisfies $\|y_i - f_i(\mathbf{x}, \{y_j\}_{j\neq i})\| = 0,~\forall i$ -- thus, $y_i^*$ is a fixed point of the system of $m$ disciplines. 
In plain terms, all disciplines must be physically consistent with each other before we observe the output $z$ of the system. 
In the context of MDO, this process of determining the state of a coupled system such that interdisciplinary compatibility is satisfied is known as multidisciplinary design analysis (MDA).

Mathematically, finding an equilibrium state involves solving a system of nonlinear equations -- analytical solutions are rarely available and, when the \(f_i\) are black-box simulation codes, they are impossible to obtain.
Instead, typical approaches apply fixed-point iteration (FPI) using iterative algorithms such as Newton's method or the Gauss--Seidel method to solve numerically for fixed points of the system \cite{mdobook,Gray2019}.
These iterative solvers are standard in popular software packages for the analysis and optimization of coupled multidisciplinary systems such as OpenMDAO~\cite{Gray2019} and GEMSEO~\cite{gemseo}.
Because FPI requires repeated evaluation of feedback-coupled components to determine converged values of system coupling variables, such methods can become computationally intractable if one or more of these components are expensive to evaluate, as is often the case with real-world applications. 
This problem is compounded when multiple evaluations of the system are required as part of an optimization loop or other outer-loop applications such as uncertainty quantification.

Surrogate models have been used extensively to reduce the cost of multidisciplinary
design analysis and optimization, ranging from global response-surface or Kriging
approximations of expensive system-level quantities
\citep{simpson2001kriging,forrester2008engineering} to architectures in which the
coupled analysis is embedded inside multidisciplinary feasible (MDF) and individual discipline feasible (IDF) MDO formulations
\citep{martins2013multidisciplinary,Yao2012MDO}.  More recent work explicitly exploits the
coupling structure.  In efficient global multidisciplinary design optimization
(EGMDO), and its constrained extension, each disciplinary solver is replaced by an
adaptively enriched Gaussian-process surrogate, after which Bayesian optimization is
used to manage the induced uncertainty in the objective and constraints
\citep{dubreuil2020towards,cardoso2024constrained}.  For parameterized MDA with
high-dimensional exchanged fields, DPOD+I replaces disciplinary solvers by
proper-orthogonal-decomposition/interpolation surrogates and enriches the most
influential discipline surrogate \citep{berthelin2022disciplinary}.  Other
approaches reduce the burden of repeated fixed-point iterations in uncertainty
propagation by learning converged coupling variables or component-level surrogates
\citep{chaudhuri2018multifidelity,jakeman2022adaptive,Hu2017mdra}, by adaptive Bayesian or
Gaussian process discipline sampling for coupled-system uncertainty analysis
\citep{ghoreishi2017adaptive,ghoreishi2021bayesian,asadi2024active}, or by
identifying weak couplings that can be removed with limited loss of output accuracy
\citep{baptista2018optimal}.  

\citet{jakeman2022adaptive} describe two categories of existing approaches for surrogates of coupled systems that exploit the known coupling structure of multidisciplinary systems.
``Class-one'' surrogates \cite{chaudhuri2018multifidelity} construct a surrogate of the coupled system by mapping system inputs to coupling variables, which are computed using FPI. 
To recover system outputs, the original system is evaluated using the converged coupling variables predicted by the surrogate.
``Class-two'' surrogates \cite{simpson2001kriging, dubreuil2020towards, kyzyurova2018gpchain, sanson2019gpchain} build individual surrogate models for each discipline using a decoupled approach, which are then assembled together into a coupled system.
To recover system outputs using a class-two surrogate, fixed-point iteration is applied to the surrogate system, with the disciplinary surrogates replacing the original expensive disciplines to obtain an equilibrium state.
Both classes of surrogates have been shown to be more efficient than black-box approaches that do not exploit the coupling structure of multidisciplinary systems. 
Of the two, class-two surrogates are comparatively more efficient due to the poor scaling of surrogate models with dimensionality, as each disciplinary surrogate is of lower order than an equivalent class-one surrogate~\cite{jakeman2022adaptive}.
However, system-level outputs of class-two surrogates may be more sensitive to the accuracy of the individual subsystem surrogates that comprise the complete model.

In contrast to existing work, the proposed approach differs in the object being learned and in how the fixed point
is recovered.  Let \(r(\x,y)=y-f(\x,y)\) denote the multidisciplinary consistency
residual, which must be driven to $0$ at the fixed point.
Instead of fitting a map from \(\x\) to the system-level output quantities of interest, or
replacing each discipline by a separate surrogate and then coupling those surrogate
components at every new design point, we place a joint multitask GP~\cite{Bonilla2007multitask} prior on the
vector residual \(r\) itself.  The learned model \(\widehat{r}(\x,y)\) shares
information across all residual components through inter-task covariance and
identifies the fixed point as a root,
\[
    \widehat{r}(\x,y^\star)=0 .
\]
Thus, the surrogate represents the implicit equilibrium manifold
\(\{(\x,y):r(\x,y)=0\}\) across the design space. A graphical illustration is provided in \Cref{fig:coupled-decoupled-models}.  This is distinct from classical
fixed-point-iteration-based MDA, which resolves the coupled equations afresh at each design. We also differ from existing surrogate-based approaches that either 
approximate and recouple the discipline maps (EGMDO/DPOD+I) or directly approximate the converged state $y^\star$, by directly learning the joint residual map whose
zeros define multidisciplinary consistency.
The benefit of learning a surrogate model over the implicit equilibrium manifold is that, once learned, fixed points for new design variables $\x$ can be cheaply predicted without restarting a fixed-point iteration.
\begin{figure}[htb!]
    \centering
    \begin{tikzpicture}[
        node distance=1.8cm,
        block/.style={rectangle, draw, minimum width=1.0cm, minimum height=0.8cm, fill=white},
        thickarrow/.style={->, thick},
        thickarrow_rev/.style={<-, thick}
    ]

        \node[block] (f1) at (0,0) {$f_1$};
        \node[block, below right=0.2cm of f1] (f2) {$f_2$};
        \node[block, below right=0.5cm of f2] (fm) {$f_m$};

        \draw[thickarrow] (f1.east) -- +(0.3,0) node[above] {$y_1$} -| (f2.north);
        \draw[thickarrow] (f2.east) -- +(0.3,0) node[above] {$y_2$} -| (fm.north);
        \draw[thickarrow] (f1.east) -| (fm.north);
        \draw[thickarrow] (f2.west) -- +(-0.3,0) node[above] {$y_2$} -| (f1.south);
        \draw[thickarrow] (fm.west) -- +(-0.3,0) node[above] {$y_m$} -| (f2.south);
        \draw[thickarrow] (fm.west)  -| (f1.south);

        \draw[dotted, very thick] ($(f2.south east)!0.1!(fm.north west)$) -- ($(f2.south east)!0.9!(fm.north west)$);

        \node[below=1.3cm of f2] (fpi) {Fixed-Point Iteration};

        \begin{pgfonlayer}{bg}
            \node[fit=(f1)(f2)(fm)(fpi), fill=magenta!20, inner sep=0.2cm, draw] (coupled system box) {};
        \end{pgfonlayer}
        \node[above=0cm of coupled system box] (coupled system box label) {\textbf{Coupled Model}};

        \node[left=0.4cm of coupled system box] (Xn) {$\mathbf{x}$};
        \draw[thickarrow] (Xn.east) -- (coupled system box.west);

        \node[below=0.4cm of coupled system box] (yn) {$\{y_{i}^*\}_{i=1}^m$};
        \draw[thickarrow] (coupled system box.south) -- (yn.north);
        
        \begin{pgfonlayer}{base}
            \node[fit=(Xn)(yn)(coupled system box)(coupled system box label), fill=black!10, minimum height=6.2cm, draw] (coupled system canvas) {};
        \end{pgfonlayer}

        \node[block, minimum width=5.2cm] at (10.5, -0.354) (r1) {$r_1=y_1-f_1(\mathbf{x}, \{y_{j}\}_{j\neq 1})$};
        \node[block, below=0.15cm of r1, minimum width=5.2cm] (r2) {$r_2=y_2-f_2(\mathbf{x}, \{y_{j}\}_{j\neq 2})$};
        \node[block, below=0.35cm of r2, minimum width=5.2cm] (rm) {$r_m=y_m-f_m(\mathbf{x}, \{y_{j}\}_{j\neq m})$};

        \draw[dotted, very thick] ($(r2.south)!0.1!(rm.north)$) -- ($(r2.south)!0.9!(rm.north)$);

        \begin{pgfonlayer}{bg}
            \node[fit=(r1)(rm), fill=cyan!20, inner sep=0.2cm, draw] (decoupled system box) {};
        \end{pgfonlayer}
        \node[above=0cm of decoupled system box] (decoupled system box label) {\textbf{Decoupled Model, Residual Form}};

        \node[left=1cm of r2] (Zn) {$\mathbf{x}, \{ y_i \}_{i=1}^m$};
        \coordinate[right=0.3cm of Zn.east] (arrow root);
        \draw[thickarrow] (arrow root) |- (r1.west);
        \draw[thickarrow] (Zn.east) -- +(-0.1,0) -- (r2.west);
        \draw[thickarrow] (arrow root) |- (rm.west);

        \node[below=0.4cm of decoupled system box] (rn) {$\{r_i\}_{i=1}^m$};
        \draw[thickarrow] (decoupled system box.south) -- (rn.north);

       \begin{pgfonlayer}{base}
            \node[fit=(Zn)(rn)(decoupled system box)(decoupled system box label), fill=black!10, minimum height=6.2cm, draw] (decoupled system canvas) {};
        \end{pgfonlayer}

        \draw[<->, double, thick] ($(coupled system canvas.east)!0.1!(decoupled system canvas.west)$) -- ($(coupled system canvas.east)!0.9!(decoupled system canvas.west)$);

        \coordinate (connector) at ($(coupled system canvas.east)!0.5!(decoupled system canvas.west)$);
        
        \node[block, above=3.3cm of connector, align=center, fill=yellow!20, draw] (eq state) {Equilibrium State $y_i^*$ satisfies: \\ $r_i(\mathbf{x},\{y_j^*\})=y^*_i-f_i(\mathbf{x},\{y^*_j\}_{j\neq i})=\mathbf{0} \quad \forall i,j \in [m]$};

        \draw[dotted, thick] (eq state.south) -- (connector);
    \end{tikzpicture}
    \caption{Conversion of coupled model to system of decoupled residuals.}
    \label{fig:coupled-decoupled-models}
\end{figure}
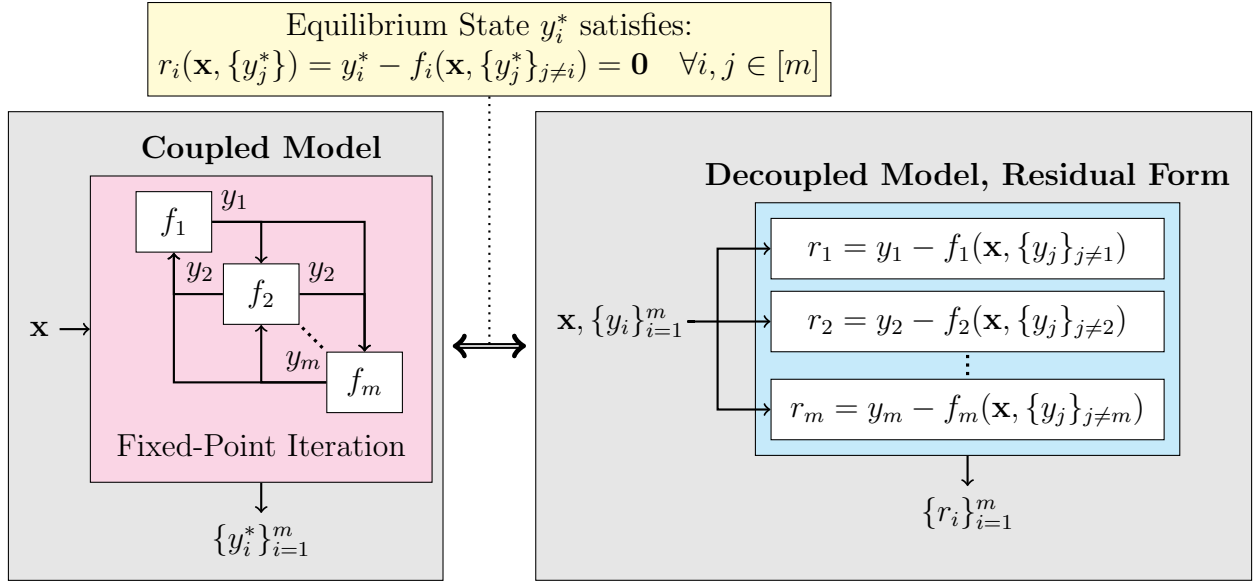

It should be noted that, conceptually, there are parallels to our approach with the broader field of contour location (CL), e.g. in reliability analysis~\cite{renganathan2023camera,renganathan2024efficient,booth2024actively} and failure probability estimation~\cite{booth2025contour,booth2025failure,ranjan2008contour,marques2018contourentropy,cole2023entropycontour,renganathan2026surrogate}.
The goal of CL is to learn the contour:
\[
    S(t)=\{\mathbf{x} \in \mathcal{X} \; | \; f(\mathbf{x})=t\}
\]
for some constant $t$ representing a failure threshold of the function $f(\x)$.
Replacing $f$ with $R$ and setting $t=0$ recovers the residual-based approach described above.
CL has benefited from a GP-like probabilistic surrogate model combined with Bayesian decision-theoretic active learning through a so-called ``acquisition function''. A notable classes of acquisition functions are based on expected improvement~\cite{jones1998efficient}, information entropy~\cite{Hennig2011,cole2023entropycontour,marques2018contourentropy,chevalier2014uncertainty,swersky2013multitask}, Thompson sampling~\cite{renganathan2025qpots, carlson2025multiobjective}, and knowledge gradient~\cite{renganathan2021lookahead, frazier2008knowledge}.
This acquisition strategy generally aims to reduce some measure of posterior uncertainty at a candidate observation site; our approach builds on this idea within the context of MDO and MDA.

We propose REMAL: residual equilibrium manifold active learning, a surrogate-based framework for multidisciplinary design analysis that learns the zero-residual manifold of a coupled system rather than the converged state map or the individual disciplinary solvers. By evaluating disciplines at prescribed design and coupling-variable inputs, the method generates training data without performing fixed-point iteration at the training points. A multitask Gaussian process then models the coupled residual components jointly, entropy-based active learning concentrates new evaluations near uncertain zero-contour regions, and equilibrium states at new design points are recovered by solving a least-residual problem on the trained surrogate. The specific contributions are:
\begin{itemize}
    \item Formulation of classical multidisciplinary design analysis as a surrogate modeling problem over the residual equilibrium manifold \( \{(\mathbf{x},\mathbf{y}):R(\mathbf{x},\mathbf{y})=\mathbf{0}\} \) -- the result is a model that enables fixed-point prediction without retraining or re-evaluating disciplines at each new design point.
    \item Construction of residual training data from decoupled disciplinary evaluations, avoiding fixed-point iteration during data generation while preserving the consistency equations that define the coupled system.
    \item A theoretical result showing that, as long as the residual grows linearly away from a local equilibrium point, the error in the surrogate-based prediction is bounded.
    \item Development of an entropy-based active learning strategy that exploits inter-residual correlations learned by a multitask GP surrogate to target uncertain zero-residual contours, while naturally balancing exploration and exploitation.
    \item Demonstration of the methodology on four coupled-system benchmarks, including feedback-coupled and feed-forward cases.
    \item A software implementation of REMAL is publicly available at \url{https://github.com/csdlpsu/surrogatemdo}.
\end{itemize}

The rest of this paper is organized as follows. We review the mathematical background in~\Cref{section:background} and present the methodology in \Cref{section:method}. The experiments are discussed in \Cref{section:numerical-experiments}, followed by concluding remarks in \Cref{section:conclusion}.

\section{Background}
\label{section:background}

\subsection{Solving for fixed points in coupled systems}
We denote by $\mathbf{y}$ the vector of all coupling variables $\mathbf{y}=\left[ y_1^T, y_2^T, \dots, y_m^T \right]^T$.
For a given coupling variable $y_i=f(\mathbf{x},\mathbf{y}_{-i})$, the corresponding coupling residual is $r_i=y_i-f(\mathbf{x},\mathbf{y}_{-i})$.
By definition, the system is solved when the residuals of all coupling variables are driven to zero~\cite{mdobook}.
Among the most common iterative methods for solving coupled models are fixed-point iteration algorithms, such as Newton's method and the Gauss-Seidel method.
These methods are included as standard nonlinear solvers in software packages such as OpenMDAO \cite{Gray2019} for the analysis and optimization of multidisciplinary systems.

The use of fixed-point iteration contributes significantly to the computational cost of analyzing and designing coupled systems. Moreover, certain fixed-point iteration methods such as Newton's method require the computation of coupled derivatives at each iteration;
for models with black-box components for which gradient information is not easily accessible, further evaluations may be required to approximate these gradients.
If any components of the coupled model are expensive to evaluate, such as those using high-fidelity computational fluid dynamics (CFD) or finite-element method (FEM) simulations, the computational cost of iteration can become prohibitively high. 

\subsection{Gaussian process models}
\label{sec:gp}
Gaussian process (GP) regression provides a probabilistic surrogate for an
unknown scalar function~\cite{rasmussen,gramacy2020surrogates}. In the context of this work, the GP is applied to the coupling
residuals rather than directly to the disciplinary maps. Let
\[
    \mathbf{u} = \left(\mathbf{x},\{y_j\}_{j=1}^m\right)
    \in
    \mathcal{X}\times \prod_{j=1}^m \mathcal{Y}_j
\]
denote the augmented input containing both design variables and candidate
coupling variables. For the \(i\)-th discipline, we define the scalar residual
\[
    r_i(\mathbf{u})
    =
    y_i - f_i\!\left(\mathbf{x},\{y_j\}_{j\neq i}\right),
    \qquad i\in[m],
\]
so that a consistent multidisciplinary state satisfies
\(r_i(\mathbf{x},\{y_j^*\}_{j=1}^m)=0\) for all \(i\in[m]\).
An independent GP model places a separate prior on each residual,
\[
    r_i(\mathbf{u})
    \sim
    \mathcal{GP}\!\left(
        \mu_i(\mathbf{u}),
        k_i(\mathbf{u},\mathbf{u}';\boldsymbol{\theta}_i)
    \right),
\]
where \(\mu_i\) is the mean function and \(k_i\) is a covariance kernel with
hyperparameters \(\boldsymbol{\theta}_i\). Given residual observations
\(\mathbf{r}_i=[r_{i,1},\ldots,r_{i,N}]^T\) at augmented inputs
\(\mathbf{U}=[\mathbf{u}_1,\ldots,\mathbf{u}_N]^T\), the training covariance is
\[
    \mathbf{K}_{y,i}
    =
    \mathbf{K}_i + \sigma_{n,i}^2\mathbf{I},
    \qquad
    (\mathbf{K}_i)_{nn'}
    =
    k_i(\mathbf{u}_n,\mathbf{u}_{n'};\boldsymbol{\theta}_i),
\]
where \(\sigma_{n,i}^2\) is an observation-noise or nugget variance. For a new
point \(\mathbf{u}_*\), the posterior distribution of \(r_i(\mathbf{u}_*)\) is
Gaussian, with mean and variance
\[
    \widehat{\mu}_i(\mathbf{u}_*)
    =
    \mu_i(\mathbf{u}_*)
    +
    \mathbf{k}_{i,*}^T
    \mathbf{K}_{y,i}^{-1}
    (\mathbf{r}_i-\boldsymbol{\mu}_i),
    \qquad
    \widehat{\sigma}_i^2(\mathbf{u}_*)
    =
    k_i(\mathbf{u}_*,\mathbf{u}_*)
    -
    \mathbf{k}_{i,*}^T
    \mathbf{K}_{y,i}^{-1}
    \mathbf{k}_{i,*}.
\]
The kernel lengthscales, process variance, mean parameters, and nugget are
typically inferred by maximizing the marginal likelihood~\cite{rasmussen}. Although this
independent construction is straightforward, it treats each residual separately
and therefore cannot exploit correlations between residual components. To address this, we propose multitask GP models, which are discussed next.

\subsection{Multitask Gaussian process models}

A multitask Gaussian process (MTGP)~\cite{Bonilla2007multitask} model extends the GP model by placing a joint prior
on multiple correlated outputs. Here, the tasks are the residual components of
the coupled multidisciplinary system. Defining
\[
    \mathbf{r}(\mathbf{u})
    =
    \left[
        r_1(\mathbf{u}),\ldots,r_m(\mathbf{u})
    \right]^T
    =
    R(\mathbf{x},\mathbf{y}),
\]
the MTGP models all residuals simultaneously over the same augmented input
\(\mathbf{u}=(\mathbf{x},\{y_j\}_{j=1}^m)\). This allows observations of one
residual component to inform predictions of another whenever the data indicate
nonzero cross-residual dependence.

We use an intrinsic coregionalization model (ICM) covariance structure for the
MTGP~\citep{Bonilla2007multitask,alvarez2012kernels}. Under this model,
\[
    \mathbf{r}(\mathbf{u})
    \sim
    \mathcal{MTGP}
    \left(
        \boldsymbol{\mu}(\mathbf{u}),
        \mathbf{K}\big((\mathbf{u},i),(\mathbf{u}',j)\big)
    \right),
\]
with covariance
\[
    \operatorname{cov}
    \left[
        r_i(\mathbf{u}),r_j(\mathbf{u}')
    \right]
    =
    B_{ij}
    k_{\mathcal{U}}
    (\mathbf{u},\mathbf{u}';\boldsymbol{\theta}_{\mathcal{U}}),
    \qquad i,j\in[m].
\]
Here, \(k_{\mathcal{U}}\) is a Mat\'ern kernel over the augmented input space and
\(\mathbf{B}\in\mathbb{R}^{m\times m}\) is a positive semidefinite
coregionalization matrix. The input kernel controls correlation between residual
evaluations at different design--coupling states, while \(\mathbf{B}\) controls
correlation between residual tasks. Its diagonal entries describe task-specific
residual variances, and its off-diagonal entries encode cross-residual
covariances. The independent GP model is recovered as the special case in which
\(\mathbf{B}\) is diagonal.

For complete observations of all residuals at \(N\) augmented inputs, stack the
training data as
\[
    \mathbf{r}
    =
    \left[
        \mathbf{r}(\mathbf{u}_1)^T,
        \ldots,
        \mathbf{r}(\mathbf{u}_N)^T
    \right]^T
    \in\mathbb{R}^{Nm}.
\]
The ICM covariance matrix can then be written compactly as
\[
    \mathbf{K}_{y}
    =
    \mathbf{K}_{\mathcal{U}}\otimes \mathbf{B}
    +
    \mathbf{I}_{N}\otimes\boldsymbol{\Sigma}_{\epsilon},
\]
where
\((\mathbf{K}_{\mathcal{U}})_{nn'}
=
k_{\mathcal{U}}(\mathbf{u}_n,\mathbf{u}_{n'};\boldsymbol{\theta}_{\mathcal{U}})\),
\(\otimes\) denotes the Kronecker product, and
\(\boldsymbol{\Sigma}_{\epsilon}\) contains task-specific nugget or noise
variances. The MTGP hyperparameters are
\[
    \boldsymbol{\theta}
    =
    \left\{
        \boldsymbol{\theta}_{\mathcal{U}},
        \mathbf{B},
        \boldsymbol{\Sigma}_{\epsilon},
        \text{mean-function parameters}
    \right\}.
\]
These are inferred by maximizing the joint log marginal likelihood
\[
    \log p(\mathbf{r}\mid\mathbf{U},\boldsymbol{\theta})
    =
    -\frac{1}{2}
    (\mathbf{r}-\boldsymbol{\mu})^T
    \mathbf{K}_{y}^{-1}
    (\mathbf{r}-\boldsymbol{\mu})
    -
    \frac{1}{2}\log|\mathbf{K}_{y}|
    -
    \frac{Nm}{2}\log(2\pi).
\]
The positive-semidefinite constraint on \(\mathbf{B}\) is commonly enforced by a
factorization such as \(\mathbf{B}=\mathbf{L}\mathbf{L}^T\), possibly with a
diagonal correction.
For a new augmented input \(\mathbf{u}_*\), the MTGP posterior over the full
residual vector is multivariate Gaussian,
\[
    \mathbf{r}(\mathbf{u}_*)\mid\mathcal{D}
    \sim
    \mathcal{N}
    \left(
        \widehat{\boldsymbol{\mu}}(\mathbf{u}_*),
        \widehat{\boldsymbol{\Sigma}}(\mathbf{u}_*)
    \right).
\]
The posterior mean gives the surrogate residual equations used to approximate
the multidisciplinary fixed point,
\[
    \widehat{\boldsymbol{\mu}}
    \left(
        \mathbf{x},\{y_j^*\}_{j=1}^m
    \right)
    =
    \mathbf{0},
\]
while the posterior covariance quantifies uncertainty in the residual vector and
retains the learned cross-residual dependencies. Thus, unlike independent GP
models, the MTGP represents the residual equilibrium manifold jointly across all
coupling residuals.

\section{Methodology}
\label{section:method}
\begin{figure}[htb!]
\centering\begin{tikzpicture}[node distance=1.8cm,block/.style={rectangle, draw, minimum width=1.0cm, minimum height=0.8cm, fill=white},thickarrow/.style={->, very thick},thickarrow_rev/.style={<-, very thick}]

\node[double copy shadow, fill=white, draw] (model) at (0,0) {$r_i \left(\mathbf{x},\{y_j\}_{j=1}^m \right)$};
\node[above=0.2cm of model, align=center] (model label) {\scriptsize Decoupled Model};
\begin{pgfonlayer}{bg}
    \node[fit=(model)(model label), fill=cyan!20, inner sep = 0.3cm, draw] (model box) {};
\end{pgfonlayer}

\node[block, above=1.4cm of model, align=center] (Zn0) 
{
    $X_{n_0}$\\
    $\left\{ {Y}_{i,n_0} \right\}_{i=1}^m$
};
\node[right=0cm of Zn0, align=center] {\textbf{Initial}\\\textbf{Training}};
\draw[thickarrow] (Zn0.south) -- (model box.north);

\node[block, fill=yellow!70, right=2.5cm of model, align=center] (gp) {Train MTGP\\Surrogate $\mathcal{S}_n$};
\draw[thickarrow] (model.east) -- 
    +(0.5,0) node[anchor=south west, align=center] (decoupled model output) {$X_n$\\$\{Y_{i,n}\}_{i=1}^m$} --
    (gp.west);
\node[below=0cm of decoupled model output, align=center] {$\{R_{i,n}\}_{i=1}^m$};
\node[right=0cm of gp.east, anchor=south west] {$\mathcal{S}_n$};

\node[block, below right=0.9cm and 0.5cm of gp.south east, anchor=north east, align=center, label=above:Active Learning] (entropy) {
    Acquire new $\mathbf{x}_{n+1}, \left\{ {y}_{i, n+1} \right\}_{i=1}^m $\\
    using entropy:\\[0.4cm]
    $
        \begin{aligned}
            \underset
                {
                    \mathbf{x} \in \mathcal{X},~ y_i \in \mathcal{Y},~\forall i \in [m]
                }
                {
                    \mathrm{argmax}
                }
            \: H\left((\mathbf{x}, \{y_i\}_{i=1}^m) \; | \; \mathcal{S}_n\right) \\
        \end{aligned}
    $
};
\draw[thickarrow] (entropy.west) -| (model box.south);
\node[below right=0cm and 0.1cm of entropy.west, anchor=north east, align=center] (out label) 
{
    $X_{n+1}$\\
    $\left\{ {Y}_{i, n+1} \right\}_{i=1}^m$\\
};
\node[above left=1.0cm and 0.3cm of out label.north, anchor=center] {\scriptsize $n \gets n+1$};

\coordinate[right=1cm of gp] (gp else);

\begin{pgfonlayer}{base}
    \node[fit=(model box)(entropy)(gp else)(Zn0), fill=black!9, draw] (training canvas) {};
    \node[anchor=north east, inner sep=0.2cm] at (training canvas.north east) (training label) {\large \textbf{TRAINING}};
\end{pgfonlayer}

\node[block, above=0.2cm of training canvas.north west, anchor=south west, align=center] (notation) 
{
$
    X_n = 
    \begin{bmatrix}  
        \mathbf{x}_1^T \\
        \mathbf{x}_2^T \\
        \vdots \\
        \mathbf{x}_n^T
    \end{bmatrix}
    \quad
$
$    
    Y_{i, n} = 
    \begin{bmatrix}  
        y_{i,1}^T \\
        y_{i,2}^T \\
        \vdots \\
        y_{i,n}^T
    \end{bmatrix}
    \quad
$
$
    R_{i, n} = 
    \begin{bmatrix}  
        r_{i,1}^T \\
        r_{i,2}^T \\
        \vdots \\
        r_{i,n}^T
    \end{bmatrix}
$};

\node[block, above right=1.5cm and 1cm of entropy.south east, anchor=south west, align=center] (sse) 
{
    Minimize SSE: \\[0.2cm]
    $
        \begin{aligned}
            \underset
                {
                    y_i \in \mathcal{Y},~\forall i \in [m]
                }
                {
                    \mathrm{argmin}
                }
            \frac{1}{2} \sum_{i=1}^{m} \hat{r}_i^2(\{y_i\}_{i=1}^m; \mathbf{x}) \\
        \end{aligned}
    $
};
\draw[thickarrow] (gp.east) -- +(2.7,0) node[above] {\scriptsize if end condition} -| (sse.north);
\draw[thickarrow] (gp else) -- +(0,-0.4) node[left] {\scriptsize else} |- (entropy.east);

\node[block, above=1cm of $(sse.north)!0.5!(sse.north east)$, label=above:Input] (xdesign) {$\mathbf{x}$};
\draw[thickarrow] (xdesign.south) -- ($(sse.north)!0.5!(sse.north east)$);

\coordinate[below right=0.15cm and 0.7cm of training label.north east] (prediction corner);

\node[block, below=0.49cm of $(sse.south)!0.5!(sse.south east)$, anchor=north] (output) {$\{y^*_i\}_{i=1}^m$};
\coordinate[below=0.2cm of output] (prediction bottom);
\draw[thickarrow] ($(sse.south)!0.5!(sse.south east)$) -- (output.north);
\node[left=0cm of output, align=center] {\textbf{Predicted}\\\textbf{Fixed Point}};

\begin{pgfonlayer}{base}
    \node[fit=(sse)(xdesign)(prediction bottom)(prediction corner), fill=black!9, draw] (prediction canvas) {};
    \node[anchor=north east, inner sep=0.2cm] at (prediction canvas.north east) (prediction label) {\large \textbf{PREDICTION}};
\end{pgfonlayer}


\coordinate[right=0.2cm of notation.north east] (ul);
\coordinate[above=0.2cm of prediction canvas.north east] (lr);
\node[block, fit=(ul)(lr), inner sep=0cm, align=center, fill=yellow!20] (openmdao2) 
    {
        \vspace*{0cm}\\
        Prediction validated\\
        against ``truth'' from\\
        \includegraphics[width=4cm]{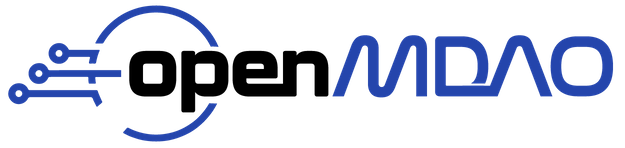}
    };

\end{tikzpicture}
\caption{Overview of residual surrogate method.}
\label{fig:block-diagram}

\end{figure}
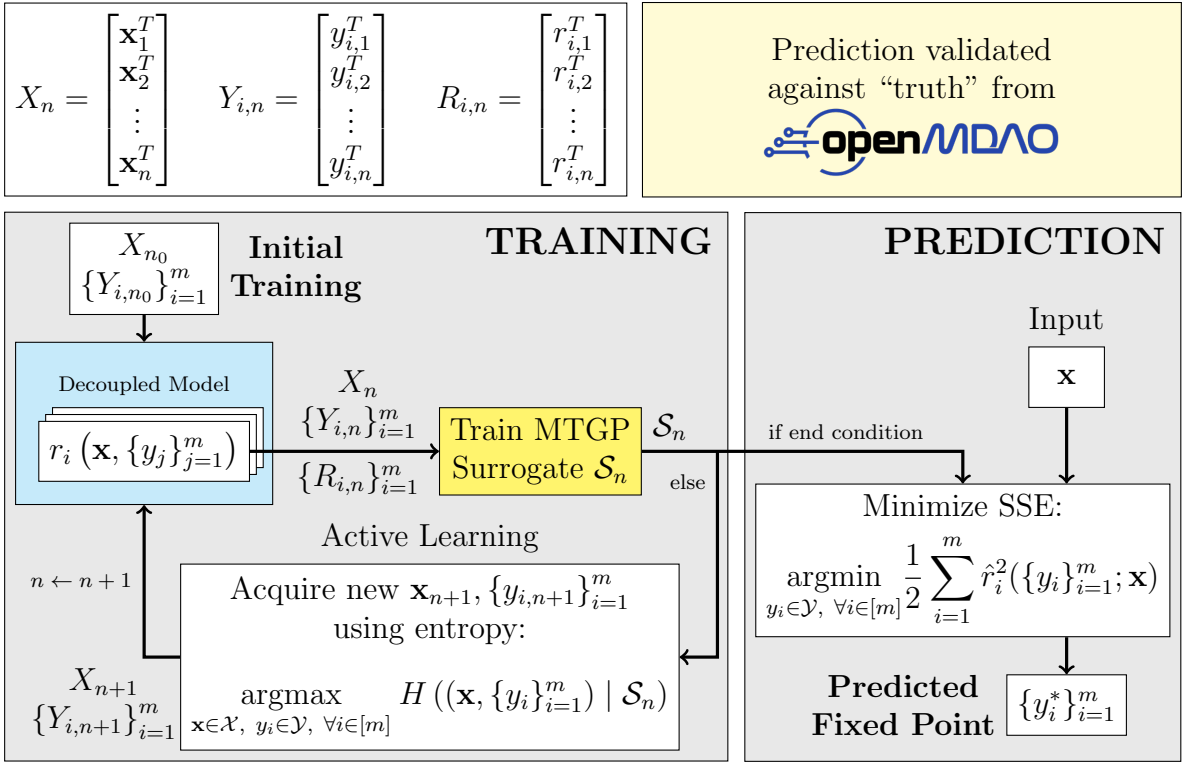
The objective of REMAL is to approximate the equilibrium state of a
coupled multidisciplinary system without repeatedly applying a fixed-point
solver to the original expensive components.  Rather than learning the
discipline maps \(f_i\) directly, or learning a map from design variables to
converged coupling variables, we learn the residual operator introduced in
Section~\ref{section:background}.  The trained surrogate is then used to recover
an equilibrium state by finding a point at which all predicted residuals vanish.

Denote $\mathcal{Y} = \prod_{i=1}^m \mathcal{Y}_i$ and $\mathcal{U} = \mathcal{X}\times\mathcal{Y}$.
For compactness, the notation below follows the scalar residual-task notation of
Section~\ref{section:background}.  If some \(y_i\) is vector-valued, the same
construction applies componentwise after stacking the corresponding scalar
residual entries.  For each discipline,
\begin{equation*}
    r_i(\mathbf{u})
    =
    y_i
    -
    f_i\!\left(
        \mathbf{x},
        \{y_j\}_{j\neq i}
    \right),
    \qquad i\in[m],
    \label{eq:method_residual_i}
\end{equation*}
and the residual vector is
\begin{equation*}
    R(\mathbf{u})
    =
    \mathbf{r}(\mathbf{u})
    =
    \left[
        r_1(\mathbf{u}),\ldots,r_m(\mathbf{u})
    \right]^T .
    \label{eq:method_residual_vector}
\end{equation*}
For a prescribed design point \(\mathbf{x}\), an equilibrium state
\(\mathbf{y}^\star(\mathbf{x})\) satisfies
$R\!\left(\mathbf{x},\mathbf{y}^\star(\mathbf{x})\right)
    =
    \mathbf{0}.$
Thus, the residual surrogate represents the implicit equilibrium manifold
\[
    \mathcal{M}
    =
    \left\{
        (\mathbf{x},\mathbf{y})\in\mathcal{U}
        :
        R(\mathbf{x},\mathbf{y})=\mathbf{0}
    \right\}.
\]
REMAL consists of four steps: construction of residual observations,
training of a multitask Gaussian process surrogate, entropy-based active
learning, and fixed-point recovery from the trained surrogate.

\subsection{Residual data generation}
\label{subsection:residual_data_generation}

Training data are generated by evaluating the disciplines in residual form.  At a
candidate augmented input
\(\mathbf{u}_n=(\mathbf{x}_n,\mathbf{y}_n)\), the coupling variables
\(\mathbf{y}_n\) are supplied as independent inputs to each discipline rather
than being obtained through fixed-point iteration.  The \(i\)-th residual
observation is therefore
\begin{equation}
    r_i(\mathbf{u}_n)
    =
    y_{i,n}
    -
    f_i\!\left(
        \mathbf{x}_n,
        \{y_{j,n}\}_{j\neq i}
    \right),
    \qquad i\in[m].
    \label{eq:method_residual_observation}
\end{equation}
Consequently, residual data can be collected without solving an MDA problem at
any training point.  Each observation requires only the corresponding
disciplinary evaluation needed to compute the residual in
\Cref{eq:method_residual_observation}.
Because the implementation uses task-indexed multitask regression, it is useful
to represent the training set as
\begin{equation}
    \mathcal{D}
    =
    \left\{
        \left(
            (\mathbf{u}_{\ell},\tau_{\ell}),
            r_{\tau_{\ell}}(\mathbf{u}_{\ell})
        \right)
    \right\}_{\ell=1}^{N_{\mathcal{D}}},
    \qquad
    \tau_{\ell}\in[m],
    \label{eq:method_task_indexed_dataset}
\end{equation}
where \(\tau_{\ell}\) is the residual-task index.  When all residuals are
observed at the same \(N\) augmented inputs, this representation is equivalent
to the stacked \(\mathbf{U}\) and \(\mathbf{r}\) notation used in
Section~\ref{section:background}; it simply expands each augmented input into
\(m\) task-indexed observations.

\subsection{Multitask residual surrogate}
\label{subsection:method_mtgp_surrogate}

The surrogate model is the MTGP residual model described in
Section~\ref{section:background}.  Conditioned on the current training set
\(\mathcal{D}\), the surrogate posterior is written as
\begin{equation}
    R(\mathbf{u})\mid \mathcal{D}
    \sim
    \mathcal{N}
    \left(
        \widehat{\boldsymbol{\mu}}_{\mathcal{D}}(\mathbf{u}),
        \widehat{\boldsymbol{\Sigma}}_{\mathcal{D}}(\mathbf{u})
    \right),
    \label{eq:method_mtgp_posterior}
\end{equation}
where $ \widehat{\boldsymbol{\mu}}_{\mathcal{D}}(\mathbf{u})
    =
    \left[
        \widehat{\mu}_{1,\mathcal{D}}(\mathbf{u}),\ldots,
        \widehat{\mu}_{m,\mathcal{D}}(\mathbf{u})
    \right]^T .$
The posterior mean provides the surrogate residual equations, and the posterior
covariance quantifies uncertainty in those residuals.

For task-indexed observations, the ICM covariance from
Section~\ref{section:background} can be written entrywise as
\begin{equation}
    k_{\theta}
    \left(
        (\mathbf{u},i),
        (\mathbf{u}',j)
    \right)
    =
    k_{\mathcal{U}}
    \left(
        \mathbf{u},\mathbf{u}';
        \boldsymbol{\theta}_{\mathcal{U}}
    \right)
    B_{ij},
    \qquad i,j\in[m],
    \label{eq:method_task_indexed_kernel}
\end{equation}
where \(k_{\mathcal{U}}\) is the input-space kernel and
\(\mathbf{B}\in\mathbb{R}^{m\times m}\) is the coregionalization matrix.  This representation is more
flexible than the full Kronecker multitask representation
\cite{Bonilla2007multitask} because it does not require every task to be
observed at every augmented input.
The MTGP hyperparameters are obtained by maximizing the exact marginal log
likelihood~\cite{rasmussen}.  
The residual observations are treated as deterministic, and a small ``nugget'' term is added to the covariance matrix
for numerical conditioning.

Before model fitting, all continuous components of \(\mathbf{u}\) are normalized
to the unit hypercube using the prescribed bounds on \(\mathcal{X}\) and
\(\mathcal{Y}\).  The task index is not normalized as a continuous coordinate;
it is used only by the task covariance.  Residual observations are standardized
before training.  All zero-residual calculations below are therefore performed
after transforming posterior means and variances back to physical residual
units, or equivalently after transforming the zero threshold into standardized
coordinates.

A practical requirement of the method is that the coupling-variable bounds used
to define \(\mathcal{Y}\) contain the equilibrium states of interest.  If these
bounds are too narrow or poorly centered, the surrogate must extrapolate when
fixed points lie outside the training domain.  Adaptive bound refinement
strategies, such as updating coupling-variable bounds from observed component
output ranges~\cite{jakeman2022adaptive}, could be incorporated into the
present framework.  In the numerical studies considered here, the bounds are set
\textit{a priori} from analytic or problem-specific information.

\subsubsection*{Initial training design}
\label{subsection:method_initial_training}

The initial training set is generated using a Latin hypercube sample
(LHS)~\cite{mckay1979lhs} over the augmented input domain \(\mathcal{U}\).  Let
\[
    \left\{
        \mathbf{u}^{(0)}_n
    \right\}_{n=1}^{N_0}
    \subset
    \mathcal{U}
\]
denote the initial LHS points.  All residual tasks are evaluated at each seed
point, giving
\begin{equation}
    \mathcal{D}_0
    =
    \left\{
        \left(
            (\mathbf{u}^{(0)}_n,i),
            r_i(\mathbf{u}^{(0)}_n)
        \right)
        :
        n=1,\ldots,N_0,\;
        i\in[m]
    \right\}.
    \label{eq:method_initial_dataset}
\end{equation}
The MTGP trained on \(\mathcal{D}_0\) provides the first posterior estimate of
the residual equilibrium manifold.  This initial design is intentionally
space-filling rather than contour-focused; subsequent samples are selected
adaptively to improve the surrogate near the zero-residual manifold.

\subsection{Entropy-based active learning}
\label{subsection:method_active_learning}

After the initial surrogate is trained, additional residual observations are
selected using an entropy-based acquisition function.  The acquisition targets
locations where the surrogate is most uncertain about whether a residual lies
above or below the zero contour.  

At any given $\mathbf{u}$, the residual $r_i(\mathbf{u})$ can be classified into one of two states: $r_i \leq 0$ or $r_i > 0$. Since the residuals are jointly learned as an MTGP and their marginals are GPs, the posterior probability that the task-\(i\) residual is below the
zero threshold is
\begin{equation}
    p_{i,\mathcal{D}}(\mathbf{u})
    =
    \mathbb{P}
    \left[
        r_i(\mathbf{u}) \leq 0
        \mid
        \mathcal{D}
    \right]
    =
    \Phi
    \left(
        -
        \frac{
            \widehat{\mu}_{i,\mathcal{D}}(\mathbf{u})
        }{
            \widehat{\sigma}_{i,\mathcal{D}}(\mathbf{u})
        }
    \right),
    \label{eq:method_zero_probability}
\end{equation}
where \(\Phi(\cdot)\) is the standard Gaussian CDF and
\(\widehat{\sigma}_{i,\mathcal{D}}^2(\mathbf{u})\) is the \(i\)-th marginal
posterior variance. The closed form of the above equation follows from the Gaussian marginal posterior of $r_i$. Treating the sign of $r_i$ as a Bernoulli random variable, the corresponding binary entropy is
\begin{equation}
    h_{i,\mathcal{D}}(\mathbf{u})
    =
    -
    p_{i,\mathcal{D}}(\mathbf{u})
    \log
    p_{i,\mathcal{D}}(\mathbf{u})
    -
    \left(
        1-p_{i,\mathcal{D}}(\mathbf{u})
    \right)
    \log
    \left(
        1-p_{i,\mathcal{D}}(\mathbf{u})
    \right).
    \label{eq:method_entropy}
\end{equation}
In practice, \(p_{i,\mathcal{D}}\) is clipped away from 0 and 1 to avoid
numerical issues in the logarithms.
At each active-learning iteration, a new candidate point is selected for each
task by solving
\begin{equation}
    \widetilde{\mathbf{u}}_i
    \in
    \operatorname*{arg\,max}_{\mathbf{u}\in\mathcal{U}}
    h_{i,\mathcal{D}}(\mathbf{u}),
    \qquad i\in[m].
    \label{eq:method_taskwise_acquisition}
\end{equation}
The maximization is inexpensive because it depends only on the MTGP posterior.
In the numerical implementation, \Cref{eq:method_taskwise_acquisition} is
solved with the L-BFGS-B algorithm~\cite{nocedal2006numerical} under the
normalized augmented-input bounds.  The selected observations are then appended
to the training set,
\begin{equation*}
    \mathcal{D}
    \leftarrow
    \mathcal{D}
    \cup
    \left\{
        \left(
            (\widetilde{\mathbf{u}}_i,i),
            r_i(\widetilde{\mathbf{u}}_i)
        \right)
        :
        i\in[m]
    \right\},
    \label{eq:method_dataset_update}
\end{equation*}
and the MTGP is retrained.  If hyperparameter optimization terminates
abnormally, the previous hyperparameters are retained and the model is
conditioned on the expanded dataset.  This stabilizes the active-learning loop
while still incorporating the new residual observations.
The active-learning loop is repeated until a prescribed evaluation budget is
reached or until the downstream fixed-point predictions satisfy the desired
accuracy criteria.  The latter criterion is assessed using the residual norm of
the surrogate prediction, as described next.

\subsection{Fixed-point prediction from the surrogate}
\label{subsection:method_fixed_point_prediction}

Once the residual surrogate has been trained, equilibrium states for new design
points are recovered without further evaluations of the expensive coupled
system.  For a prescribed design point \(\mathbf{x}_0\), the zero-residual
level set is approximated by
solving a nonlinear least squares problem using the MTGP posterior mean:
\begin{equation}
    \widehat{\mathbf{y}}^\star_{\mathcal{D}}(\mathbf{x}_0)
    \in
    \operatorname*{arg\,min}_{\mathbf{y}}
    J_{\mathcal{D}}(\mathbf{y};\mathbf{x}_0),
    \qquad
    J_{\mathcal{D}}(\mathbf{y};\mathbf{x}_0)
    =
    \frac{1}{2}
    \left\|
        \widehat{\boldsymbol{\mu}}_{\mathcal{D}}
        \left(
            \mathbf{x}_0,\mathbf{y}
        \right)
    \right\|_2^2 .
    \label{eq:method_fixed_point_objective}
\end{equation}
Equivalently, \Cref{eq:method_fixed_point_objective} seeks the intersection
of the predicted zero-level sets
\begin{equation}
    \widehat{S}_{i,\mathcal{D}}(0;\mathbf{x}_0)
    =
    \left\{
        \mathbf{y}
        :
        \widehat{\mu}_{i,\mathcal{D}}
        \left(
            \mathbf{x}_0,\mathbf{y}
        \right)
        =
        0
    \right\},
    \qquad i\in[m].
    \label{eq:method_predicted_zero_sets}
\end{equation}
A candidate solution is accepted as an approximate equilibrium state if
\begin{equation}
    \left\|
        \widehat{\boldsymbol{\mu}}_{\mathcal{D}}
        \left(
            \mathbf{x}_0,
            \widehat{\mathbf{y}}^\star_{\mathcal{D}}(\mathbf{x}_0)
        \right)
    \right\|_2
    \leq
    \varepsilon_R,
    \label{eq:method_acceptance_criterion}
\end{equation}
for a prescribed residual tolerance \(\varepsilon_R\).  If no candidate satisfies
\Cref{eq:method_acceptance_criterion}, the optimizer returns only a
least-residual point and the surrogate is not considered to have identified a
fixed point for that design. \Cref{prop:stability bound} shows that, under mild assumptions, the error in the predicted fixed point is bounded by a quantity involving $\varepsilon_R$.

The objective in \Cref{eq:method_fixed_point_objective} is generally
nonconvex, and multiple fixed points may exist.  Therefore, the optimization is
initialized from multiple coupling-variable guesses within a chosen prediction
search region.  In the numerical examples, the objective is minimized using the Newton conjugate-gradient trust-region method~\cite{nocedal2006numerical}.  Gradients and
Hessian-vector products are computed by automatic differentiation of
\Cref{eq:method_fixed_point_objective}.  If a
system-level output is required after the fixed point is recovered, it is
evaluated as
\[
    \widehat{z}(\mathbf{x}_0)
    =
    g
    \left(
        \mathbf{x}_0,
        \{\widehat{y}^{\star}_{i,\mathcal{D}}(\mathbf{x}_0)\}_{i=1}^m
    \right).
\]

\subsection*{Validation against classical fixed-point iteration}
\label{subsection:method_validation}

The surrogate-predicted fixed points are validated against conventional MDA
solutions computed using {OpenMDAO}~\cite{Gray2019}.  For each test design
\(\mathbf{x}_q\), OpenMDAO is used to compute a reference equilibrium state
\(\mathbf{y}^{\star}_{\mathrm{FPI}}(\mathbf{x}_q)\) using the same disciplinary
equations as the surrogate-data generator.  When applicable, the OpenMDAO
analyses use \texttt{NonlinearBlockGS} as the nonlinear fixed-point solver and
\texttt{DirectSolver} as the associated linear solver.

The fixed point prediction error is measured by the Euclidean distance between
the surrogate prediction and the reference fixed point,
\begin{equation}
    e_y(\mathbf{x}_q)
    =
    \left\|
        \widehat{\mathbf{y}}^\star_{\mathcal{D}}(\mathbf{x}_q)
        -
        \mathbf{y}^{\star}_{\mathrm{FPI}}(\mathbf{x}_q)
    \right\|_2 .
    \label{eq:method_validation_error}
\end{equation}
These OpenMDAO solutions are used only for validation in the numerical
experiments; they are not used to train the MTGP surrogate or to select
active learning samples. The overall method is summarized in \Cref{alg:residual_mtgp}.

\begin{algorithm}[t]
\caption{REMAL: Residual Equilibrium Manifold Active Learning}
\label{alg:residual_mtgp}
\begin{algorithmic}[1]
\REQUIRE Design domain $\mathcal{X}$, coupling domain $\mathcal{Y}$, residual
definitions in \Cref{eq:method_residual_i}, initial size $N_0$, active-learning
budget $N_{\mathrm{AL}}$, tolerance $\varepsilon_R$.
\ENSURE Trained residual surrogate and predicted fixed points
$\widehat{\mathbf{y}}^\star_{\mathcal{D}}(\mathbf{x})$.

\STATE Draw an initial LHS
$\{\mathbf{u}^{(0)}_n=(\mathbf{x}^{(0)}_n,\mathbf{y}^{(0)}_n)\}_{n=1}^{N_0}
\subset \mathcal{X}\times\mathcal{Y}$.
\STATE Evaluate the disciplinary residuals at the seed points using
\Cref{eq:method_residual_observation}, without performing fixed-point
iteration.
\STATE Form the task-indexed dataset $\mathcal{D}_0$ as in
\Cref{eq:method_initial_dataset}; normalize continuous inputs and
standardize residual observations.

\STATE Train the multitask GP residual surrogate
$R(\mathbf{u})|\mathcal{D}\sim
\mathcal{N}(\widehat{\boldsymbol{\mu}}_{\mathcal{D}},
\widehat{\boldsymbol{\Sigma}}_{\mathcal{D}})$ using the ICM kernel in
\Cref{eq:method_task_indexed_kernel} and the associated marginal likelihood.

\FOR{$k=1,\ldots,N_{\mathrm{AL}}$}
    \FOR{each residual task $i=1,\ldots,m$}
        \STATE Compute the zero-contour probability
        $p_{i,\mathcal{D}}(\mathbf{u})$ and entropy
        $h_{i,\mathcal{D}}(\mathbf{u})$ using
        Eqs.~\eqref{eq:method_zero_probability}--\eqref{eq:method_entropy}.
        \STATE Select
        $\widetilde{\mathbf{u}}_i
        \in \arg\max_{\mathbf{u}\in\mathcal{U}}
        h_{i,\mathcal{D}}(\mathbf{u})$ as in
        \Cref{eq:method_taskwise_acquisition}.
        \STATE Evaluate only the selected residual
        $r_i(\widetilde{\mathbf{u}}_i)$.
    \ENDFOR
    \STATE Augment $\mathcal{D}$ using \Cref{eq:method_dataset_update}
    and retrain, or recondition with previous hyperparameters if retraining fails.
\ENDFOR

\STATE For a new design $\mathbf{x}_0$, recover the equilibrium prediction by
solving the least-residual problem in
\Cref{eq:method_fixed_point_objective} from multiple initial guesses.
\STATE Accept the solution as a surrogate fixed point if the residual criterion
in \Cref{eq:method_acceptance_criterion} is satisfied; otherwise report the
least-residual point only.
\STATE If needed, evaluate the system output
$\widehat{z}(\mathbf{x}_0)=
g(\mathbf{x}_0,\{\widehat{y}^{\star}_{i,\mathcal{D}}(\mathbf{x}_0)\}_{i=1}^m)$.
\end{algorithmic}
\end{algorithm}

\subsection{Theoretical properties}

The least-squares optimization step in \Cref{eq:method_fixed_point_objective} results in a bounded prediction error on the fixed point.  The following proposition formalizes this claim.

\begin{proposition}[A local error bound]\label{prop:stability bound}
    Fix a design point \(\x_0\), and define
\[
R_{\x_0}(y) = R(\x_0,y), \qquad \widehat R_D(y)=\widehat \mu_D(\x_0,y).
\]
Let \(y^\star\) be a locally unique equilibrium satisfying
\(R_{\x_0}(y^\star)=0\). Assume that there exist constants
\(\rho>0\) and \(\alpha>0\) such that
\[
\|R_{\x_0}(y)\|_2 \ge \alpha \|y-y^\star\|_2,
\qquad
\forall y\in B_\rho(y^\star),
\]
where $B_\rho$ is a ball of radius $\rho$ centered at $y^\star$.
Let
\[
\delta_D =
\sup_{y\in B_\rho(y^\star)}
\|\widehat R_D(y)-R_{\x_0}(y)\|_2 .
\]
If the surrogate-predicted point \(\widehat y_D^\star\in B_\rho(y^\star)\)
satisfies
$\|\widehat R_D(\widehat y_D^\star)\|_2 \le \varepsilon_R,$
then
\[
\|\widehat y_D^\star-y^\star\|_2
\le
\frac{\delta_D+\varepsilon_R}{\alpha}.
\]
\end{proposition}

\begin{proof}
    Using the local residual stability assumption and \(R_{\x_0}(y^\star)=0\),
\[
\alpha \|\widehat y_D^\star-y^\star\|_2
\le
\|R_{\x_0}(\widehat y_D^\star)\|_2 .
\]
By adding and subtracting the surrogate residual and applying the Cauchy--Schwarz inequality, we obtain
\[
\|R_{\x_0}(\widehat y_D^\star)\|_2
\le
\|R_{\x_0}(\widehat y_D^\star)-\widehat R_D(\widehat y_D^\star)\|_2
+
\|\widehat R_D(\widehat y_D^\star)\|_2
\le
\delta_D+\varepsilon_R .
\]
Combining the above two inequalities gives the stated result.
\end{proof}
The proposition shows that the fixed-point error is controlled by two
quantities: the residual-surrogate error \(\delta_D\) near the true equilibrium
and the residual tolerance \(\varepsilon_R\) used in the least-residual solve.
The constant \(\alpha\) measures the conditioning of the residual root. Small
\(\alpha\) corresponds to nearly tangent or weakly intersecting residual
contours, in which case a small residual error can produce a larger error in
the recovered coupling variables. This is consistent with the contour-based
interpretation used throughout the numerical examples. 

\section{Numerical experiments}
\label{section:numerical-experiments}

We evaluate the proposed residual MTGP approach on four coupled-system examples. 
The examples include two low-dimensional feedback-coupled systems, a feed-forward turbine model, and a modified turbine model with added feedback coupling. 
Together, these cases test whether the residual surrogate can learn equilibrium manifolds across different coupling structures and input dimensions.

For each experiment, an initial training set
\[
    S_0
    =
    \left\{
        (\mathbf{x}^{(j)},\mathbf{y}^{(j)})
    \right\}_{j=1}^{n_0}
    \subset \mathcal{X}\times\mathcal{Y}
\]
is generated by Latin hypercube sampling over the augmented input domain. 
All \(m\) residual components are evaluated at each seed point, giving \(n_0m\) initial residual observations. 
These observations initialize the MTGP surrogate. 
The surrogate uses the ICM multitask covariance structure~\cite{Bonilla2007multitask}. 
After initialization, the model is updated for \(n_{\mathrm{AL}}\) active-learning iterations using the entropy acquisition rule described in Section~\ref{subsection:method_active_learning}. 
At each iteration, one new residual observation is acquired for each task, so the total number of residual evaluations is
\[
    N_{\mathrm{eval}}
    =
    m(n_0+n_{\mathrm{AL}}).
\]

To provide a baseline for the entropy-based acquisition strategy, we also choose acquisitions uniformly at random for comparison.
Thus, the entropy and random acquisition strategies use the same residual-evaluation budget in each comparison.

Surrogate accuracy is monitored using a fixed set of test design points $\left\{
        \mathbf{x}_q
    \right\}_{q=1}^{N_{\text{test}}}
    \subset \mathcal{X},$
chosen before training. 
At each active-learning iteration, the current surrogate is used to predict the equilibrium coupling variables at all the $N_{\text{test}}$ points. 
Reference equilibrium states are computed independently using OpenMDAO with nonlinear block Gauss-Seidel fixed-point iteration whenever the model contains feedback coupling. 
Crucially, the surrogate is not retrained separately for each test point; a single surrogate iterate is used for all $N_{\text{test}}$ predictions.

Unless otherwise stated, prediction error is reported as the Euclidean norm between the surrogate-predicted fixed point and the OpenMDAO reference solution after scaling each coupling variable to its prescribed bounds. 
This normalized error makes the reported convergence histories comparable across coupling variables with different physical units and magnitudes. 
Each experiment is repeated 20 times with different initial LHS designs. 
The convergence plots report the mean error over these 20 repetitions; the shaded envelopes show the minimum and maximum error at each iteration. 
Entropy-based acquisition is shown in blue, random acquisition in red, and OpenMDAO fixed-point histories, where applicable, are shown in black with residual evaluations accumulated across the $N_{\text{test}}$ points.

\subsection{Satellite problem}
\label{sec:satellite}
The first example is adapted from~\citet{Sankararaman2012satellite}. 
The original model contained three components, but only two were coupled through a feedback loop. 
Here, the third component is omitted because it only computes an additional output state and does not affect the coupled solution. 
The resulting model has two coupling variables,
\[
    y_1=u_{12}\in[6,12],
    \qquad
    y_2=u_{21}\in[6,20],
\]
and a five-dimensional design vector \(\mathbf{x}\in[0,2]^5\); the $\textrm N^2$ diagram is shown in \Cref{fig:satellite_block_diagram}. 
The design variables are assumed to be independent. 
The coupling equations and residuals are
\begin{gather*}
    y_1 = u_{12} = x_1^2+2x_2-x_3+2\sqrt{u_{21}},
    \label{eq:coupling-u12} \\
    y_2 = u_{21} = x_1x_4+x_4^2+x_5+u_{12},
    \label{eq:coupling-u21} \\
    r_1
    =
    u_{12}
    -
    \left(
        x_1^2+2x_2-x_3+2\sqrt{u_{21}}
    \right),
    \label{eq:residual-u12} \\
    r_2
    =
    u_{21}
    -
    \left(
        x_1x_4+x_4^2+x_5+u_{12}
    \right).
    \label{eq:residual-u21}
\end{gather*}
For this example, \(n_0=8\) and \(n_{\mathrm{AL}}=142\), giving 150 evaluations per residual and 300 total residual evaluations.

\begin{figure}[h]
    \centering
    \begin{tikzpicture}[
        node distance=2.2cm,
        block/.style={rectangle, draw, minimum height=1.2cm, minimum width=3cm, align=center, fill=white},
        smallblock/.style={rectangle, draw, minimum height=0.9cm, minimum width=1.2cm, align=center},
        thickarrow/.style={->, thick}
    ]
        
        \node[block] (c1) {Component 1\\[2pt] $y_1=f_1(x_1,x_2,x_3,y_2)$};
        \node[block, below right=0.5cm and 0.5cm of c1] (c2) {Component 2\\[2pt] $y_2=f_2(x_1,x_4,x_5,y_1)$};

        \node[smallblock, above=0.5cm of c1] (x1) {$x_1, x_2, x_3$};
        \draw[thickarrow] (x1.south) -- (c1.north);
        \node[smallblock] at (x1 -| c2) (x2) {$x_1, x_4, x_5$};
        \draw[thickarrow] (x2.south) -- (c2.north);

        \draw[thickarrow] (c1.east) -| (c2.north);
        \draw[thickarrow] (c2.west) -| (c1.south);
        
    \end{tikzpicture}
    \caption{$\textrm{N}^2$ diagram of the satellite problem.}
    \label{fig:satellite_block_diagram}
\end{figure}
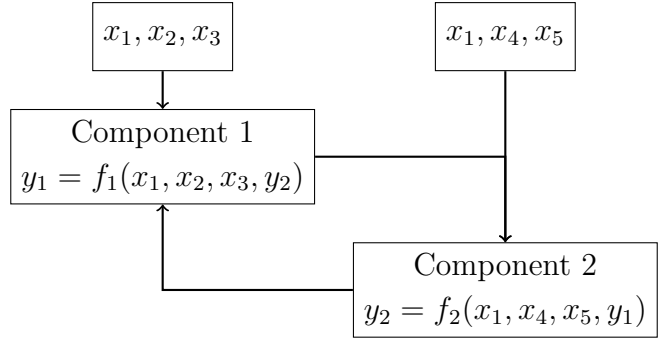

Figure~\ref{fig:satellite-sweep} compares the surrogate-predicted zero-residual contours with the true residual contours for $N_{\text{test}}=10$ test design points selected by Latin hypercube sampling. 
The OpenMDAO fixed points are also shown. 
The predicted and true contours are visually close for most inputs, indicating that the surrogate has learned the local geometry of the zero-residual manifold. 
For test inputs 4 and 10, the contour intersections occur outside the prescribed coupling-variable bounds; these cases illustrate the sensitivity of the method to the choice of \(\mathcal{Y}\), since accurate prediction requires the equilibrium state to lie within, or sufficiently near, the training domain.
\begin{figure}[htb!]
    \centering
    \includegraphics[width=1.0\linewidth]{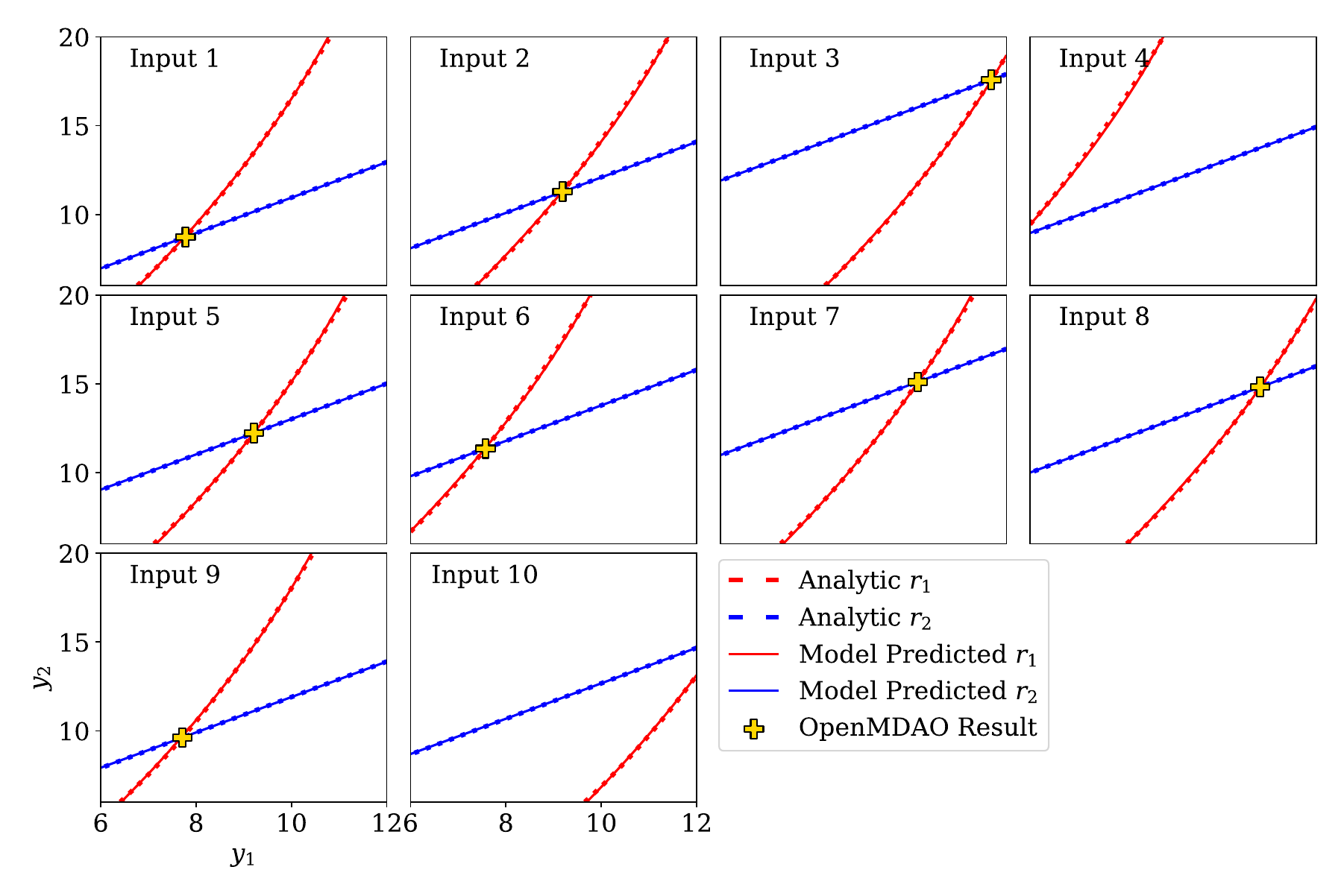}
    \caption{Comparison of zero-residual contours for the satellite system at ten design-variable inputs. OpenMDAO fixed points are also shown. Some contour intersections occur outside the prescribed bounds on \(y_1\) and \(y_2\), as in inputs 4 and 10, which highlights the importance of selecting coupling-variable bounds that contain the equilibrium states of interest.}
    \label{fig:satellite-sweep}
\end{figure}

Figure~\ref{fig:satellite-convergence} shows the normalized fixed-point error at the 10 test inputs. 
Entropy-based acquisition reduces the mean error to the order of \(10^{-4}\) after 300 total residual evaluations and continues to improve over the allotted budget. 
It also substantially outperforms random acquisition for the same number of evaluations. 
After approximately eight consecutive test-point analyses, the cumulative OpenMDAO residual-evaluation count exceeds the number of evaluations needed for the surrogate to reach an average error of approximately \(10^{-3}\), although OpenMDAO converges each individual fixed-point solve to its specified solver tolerance.

\begin{figure}[htb!]
    \centering
    \begin{subfigure}{\linewidth}
        \centering
        \includegraphics[width=1.0\linewidth]{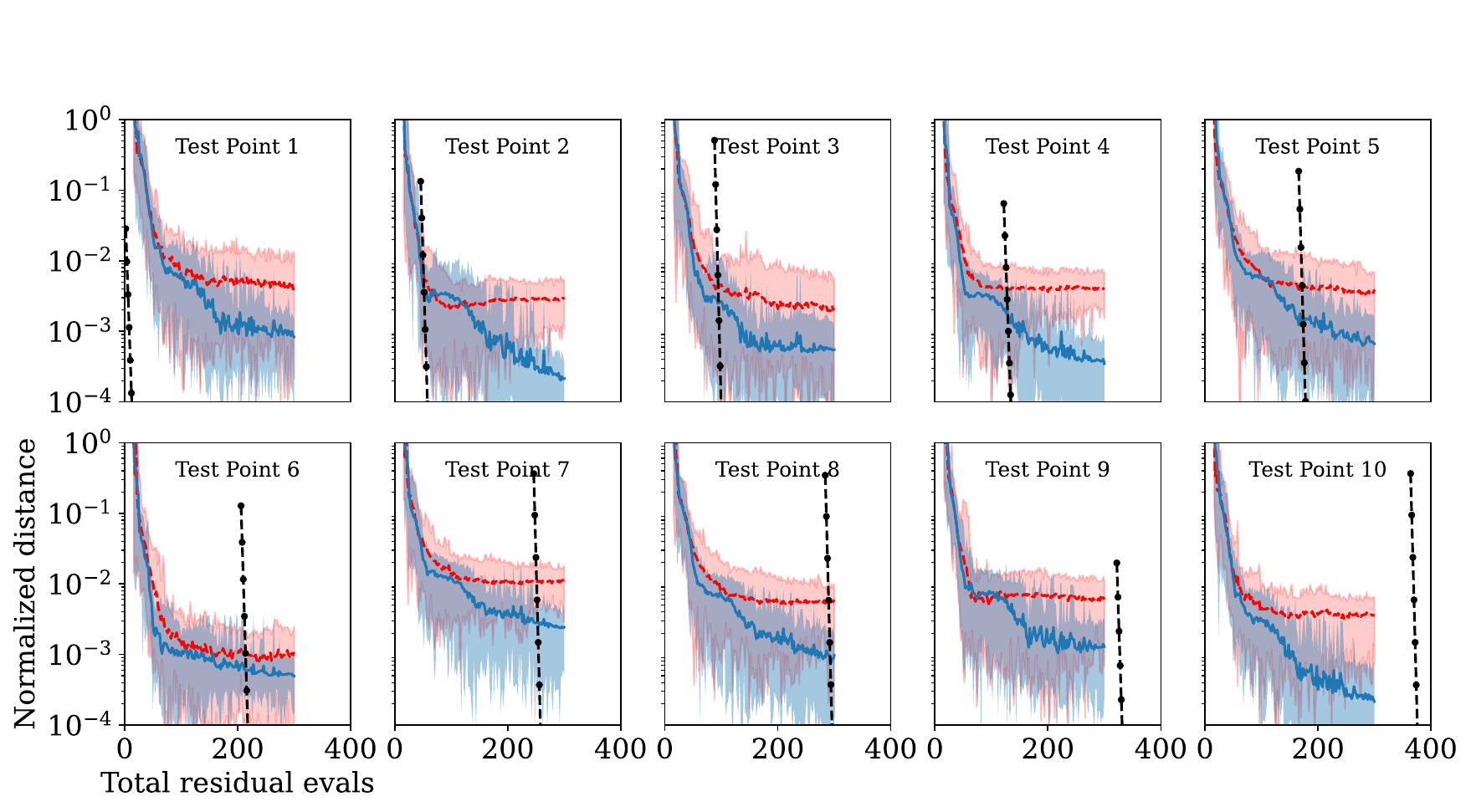}
    \end{subfigure}
    \par\medskip
    \begin{subfigure}{1\linewidth}
        \centering
        \includegraphics[width=\linewidth]{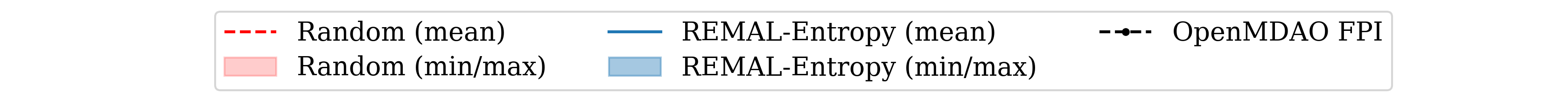}
    \end{subfigure}
    \caption{Convergence history for the satellite problem: normalized distance from the OpenMDAO reference fixed point at 10 test inputs over 142 active-learning iterations, starting from 8 initial seed points.}
    \label{fig:satellite-convergence}
\end{figure}

\subsection{Aerostructures problem}

The second example is adapted from~\citet{ghoreishi2020aerostructures}. 
It represents an aerostructural system consisting of an elastically constrained lifting body in two-dimensional flow. 
A notional diagram is shown in \Cref{fig:aeroelastic_diagram}. 
The two coupling variables are steady-state lift,
\[
    y_1=L\in[0,300],
\]
and pitch angle,
\[
    y_2=\phi\in[-\pi/2,\pi/2].
\]
The scalar design variable \(B\in[0,300)\) is the span of the lifting surface in centimeters.

\begin{figure}[htb!]
    \centering
    \begin{subfigure}{.5\textwidth}
        \includegraphics[width=1\linewidth]{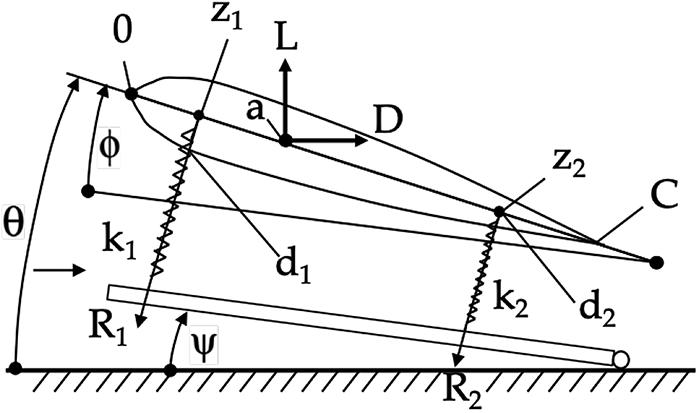}
    \caption{Aerostructural airfoil setup~\cite{ghoreishi2020aerostructures}.}
    \label{fig:aeroelastic_diagram}
    \end{subfigure}%
    \begin{subfigure}{.5\textwidth}
        \centering
            \begin{tikzpicture}[
        node distance=2.2cm,
        block/.style={rectangle, draw, minimum height=1.2cm, minimum width=3cm, align=center, fill=white},
        smallblock/.style={rectangle, draw, minimum height=0.9cm, minimum width=1.2cm, align=center},
        thickarrow/.style={->, thick}
    ]
        
        \node[block] (c1) {Lift\\[2pt] $y_1=f_1(B,y_2)$};
        \node[block, below right=0.5cm and 0.5cm of c1] (c2) {Pitch Angle\\[2pt] $y_2=f_2(y_1)$};

        \node[smallblock, above=0.5cm of c1] (x1) {$B$};
        \draw[thickarrow] (x1.south) -- (c1.north);

        \draw[thickarrow] (c1.east) -| (c2.north);
        \draw[thickarrow] (c2.west) -| (c1.south);
        
    \end{tikzpicture}
    \caption{$\textrm{N}^2$ diagram.}
    \label{fig:aerostrux_block_diagram}
    \end{subfigure}
    \caption{Aerostructural problem setup and coupling structure.}
    \label{fig:aerostructural_overview}
\end{figure}
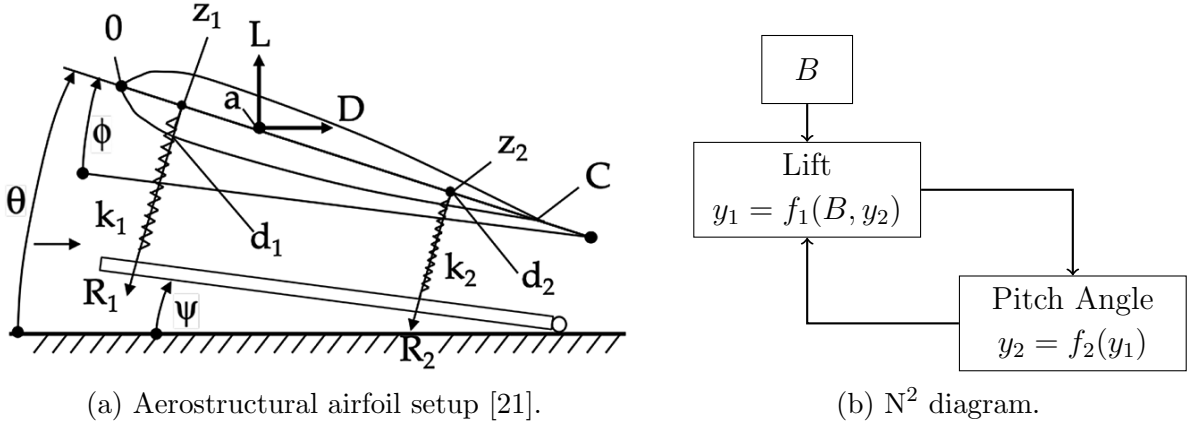

The disciplinary equations are
\begin{gather*}
    L
    =
    f_1(\phi)
    =
    qBC
    \left(
        2\pi(\phi+\psi)
        +
        r
        \left[
            1
            -
            \cos
            \left(
                \frac{\pi(\phi+\psi)}{2\theta}
            \right)
        \right]
    \right),
    \label{eq:coupling-lift} \\
    \phi
    =
    f_2(L)
    =
    \left(
        \frac{L}{k_1(1+p)}
        -
        \frac{Lp}{k_2(1+p)}
    \right)
    \frac{1}{C(z_2-z_1)}.
    \label{eq:coupling-phi}
\end{gather*}
The residuals are
\begin{gather*}
    r_1 = L-f_1(\phi), \\
    r_2 = \phi-f_2(L).
\end{gather*}
For this example, \(n_0=8\) and \(n_{\mathrm{AL}}=92\), giving 100 evaluations per residual and 200 total residual evaluations.

Figure~\ref{fig:aerostrux-sweep} shows surrogate-predicted and true zero-residual contours for $N_{\text{test}}=10$ values of \(B\). 
The selected values are evenly spaced over the design-variable domain and are shown in increasing order. 
The OpenMDAO fixed points are also shown. 
As in the satellite case, the predicted contours agree closely with the true contours over the test set.
\begin{figure}[htb!]
    \centering
    \includegraphics[width=1.0\linewidth]{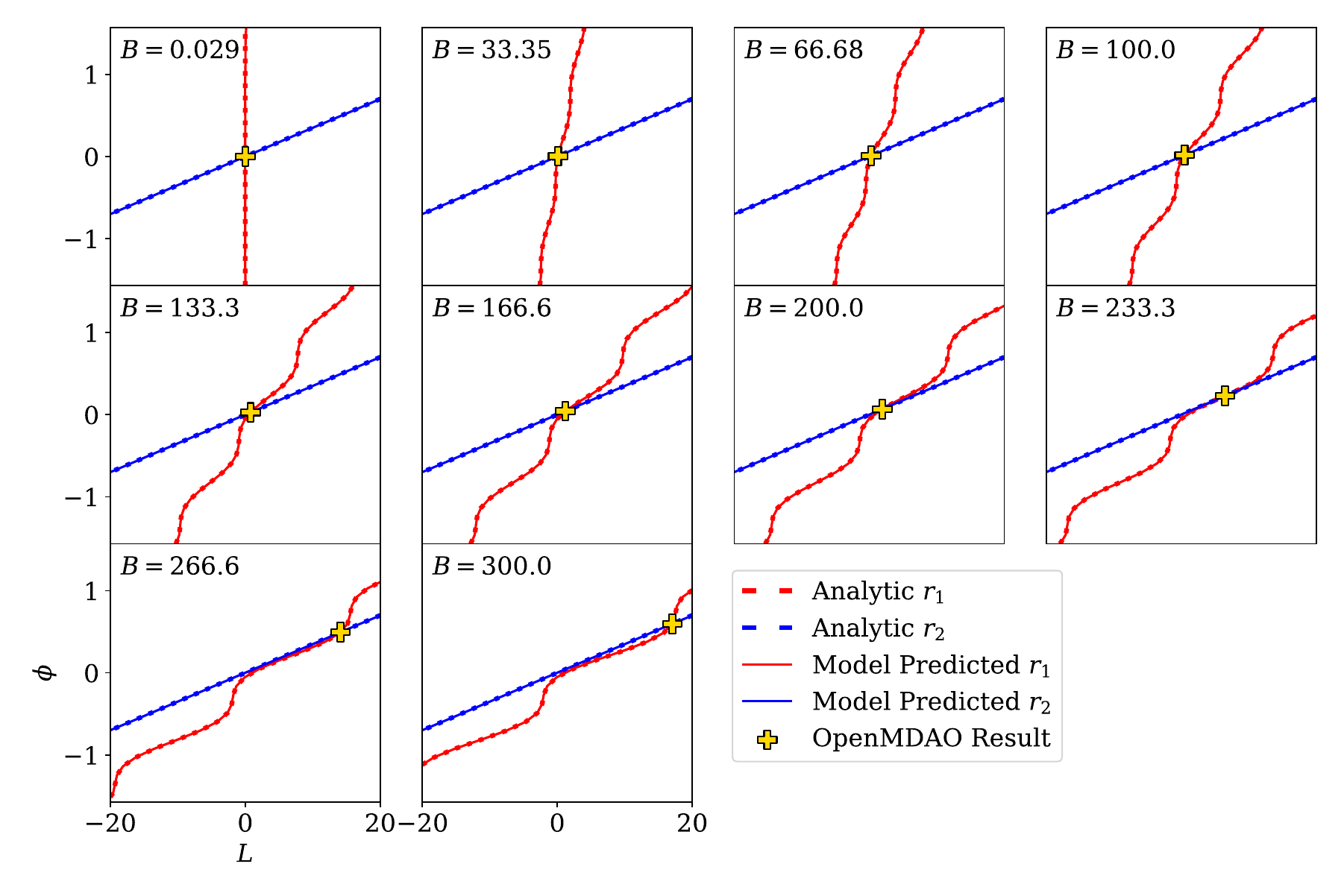}
    \caption{Comparison of zero-residual contours for the aerostructures system at ten design-variable inputs. OpenMDAO fixed points are also shown.}
    \label{fig:aerostrux-sweep}
\end{figure}

\begin{figure}[htb!]
    \centering
    \begin{subfigure}{\linewidth}
        \centering
        \includegraphics[width=1.0\linewidth]{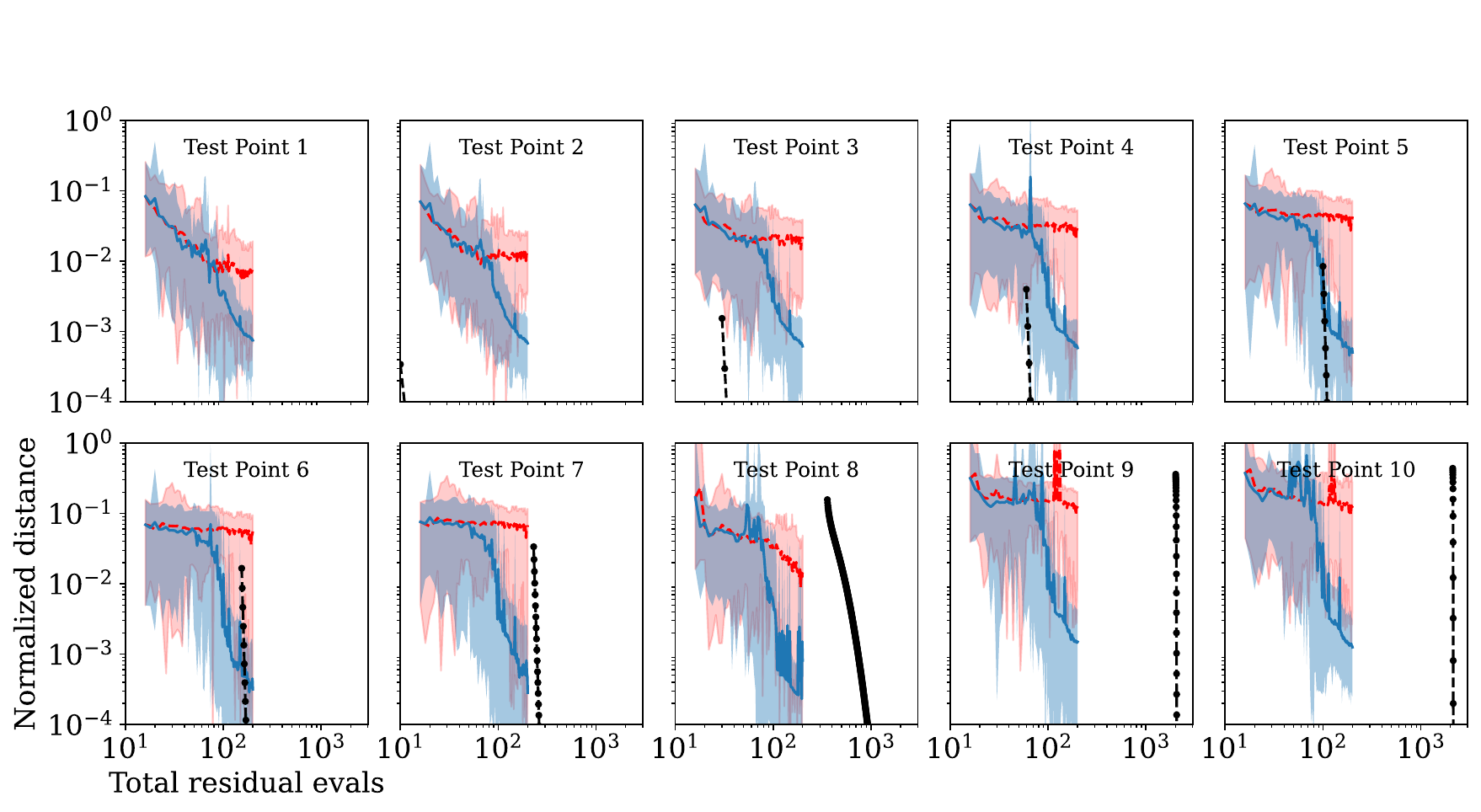}
    \end{subfigure}
    \par\medskip
    \begin{subfigure}{1\linewidth}
        \centering
        \includegraphics[width=1\linewidth]{figures/legend.png}
    \end{subfigure}
    \caption{Convergence history for the aerostructural problem: normalized distance from the OpenMDAO reference fixed point at $N_{\text{test}}=10$ test inputs over 92 active-learning iterations, starting from 8 initial seed points.}
    \label{fig:aerostrux-convergence}
\end{figure}

The convergence histories in \Cref{fig:aerostrux-convergence} follow the same format as in the satellite example. 
After 200 total residual evaluations, entropy-based active learning achieves a mean normalized error on the order of \(10^{-4}\) and substantially outperforms random acquisition. 
After approximately six consecutive test-point analyses, the cumulative number of residual evaluations used by OpenMDAO exceeds the number required for the surrogate to reach an average error near \(10^{-3}\), although OpenMDAO still converges each individual solve to the specified solver tolerance.

The eighth test point, corresponding to \(B=233.33\), is comparatively difficult for OpenMDAO. 
Inspection of the residual contours in \Cref{fig:aerostrux-sweep} shows that the two true contours are nearly parallel at their intersection, which slows nonlinear block Gauss--Seidel convergence. 
Figure~\ref{fig:aerostrux-spike} shows that the number of OpenMDAO iterations increases sharply as \(B\) approaches this value. 
The surrogate prediction error does not show a comparable localized degradation in \Cref{fig:aerostrux-convergence}.

\begin{figure}[hb!]
    \centering
    \includegraphics[width=0.5\linewidth]{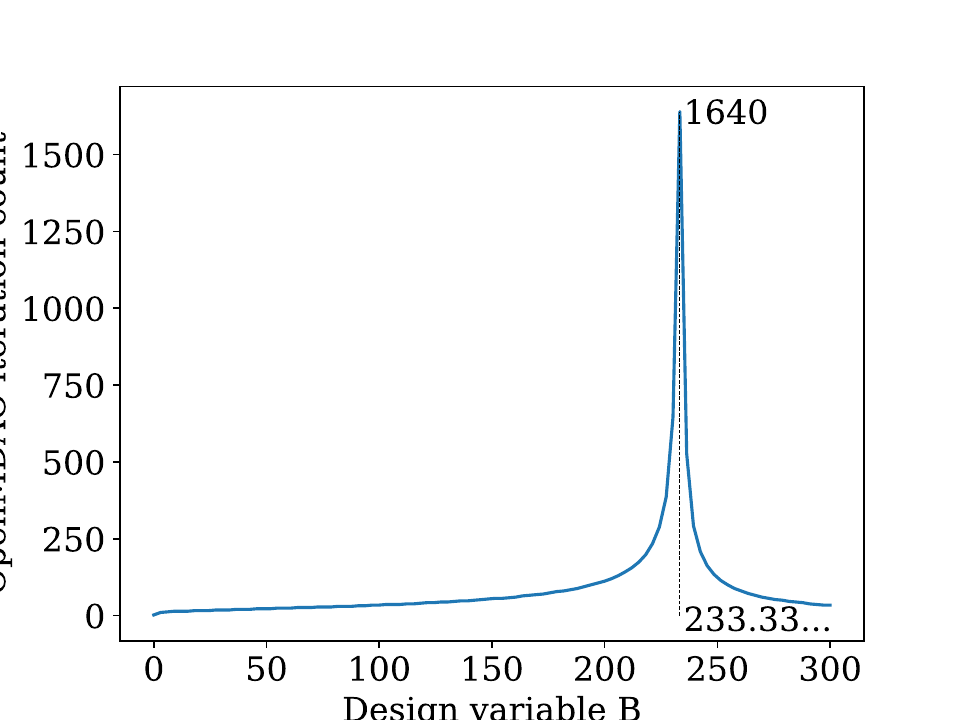}
    \caption{Number of nonlinear block Gauss--Seidel iterations required for OpenMDAO to converge the aerostructures model as a function of the design variable \(B\).}
    \label{fig:aerostrux-spike}
\end{figure}

\subsection{Finite-element gas-turbine heat-transfer and economics}
\begin{figure}[h!]
    \centering
    \begin{tikzpicture}[
        circ/.style={circle, draw, minimum size = 0.5cm, inner sep = 0pt, fill={rgb:black,1;white,15}, font=\bfseries},
    ]
        \draw (0, 0) node[inner sep=0] {\includegraphics[width=0.88\linewidth]{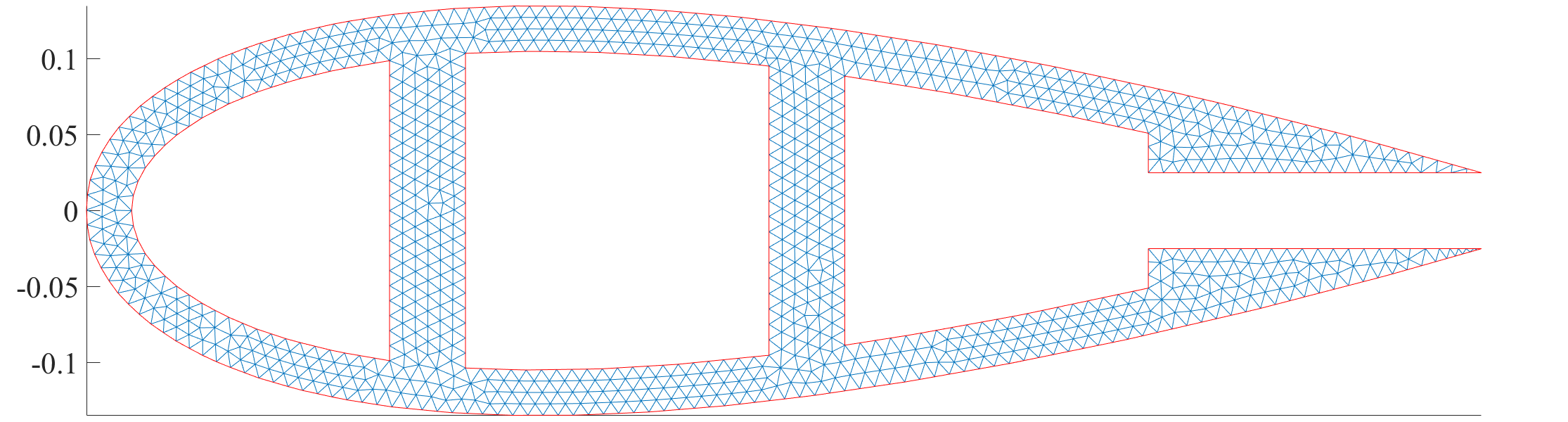}};
        
        \node at (-4.8, 0) [circ] (label1) {1};
        \draw (label1.east) -- (-3.65, 0);

        \node at (-1.6, 0) [circ] (label2) {2};
        \draw (label2.east) -- (-0.13, 0);

        \node at (1.8, 0) [circ] (label3) {3};
        \draw (label3.east) -- (3, -0.8);

        \node at (-6, 1.5) [circ] (label4) {4};
        \draw (label4.east) -- (-5.25, 1.4);
    \end{tikzpicture}
    \hspace*{0.42cm}\includegraphics[width=0.9\linewidth]{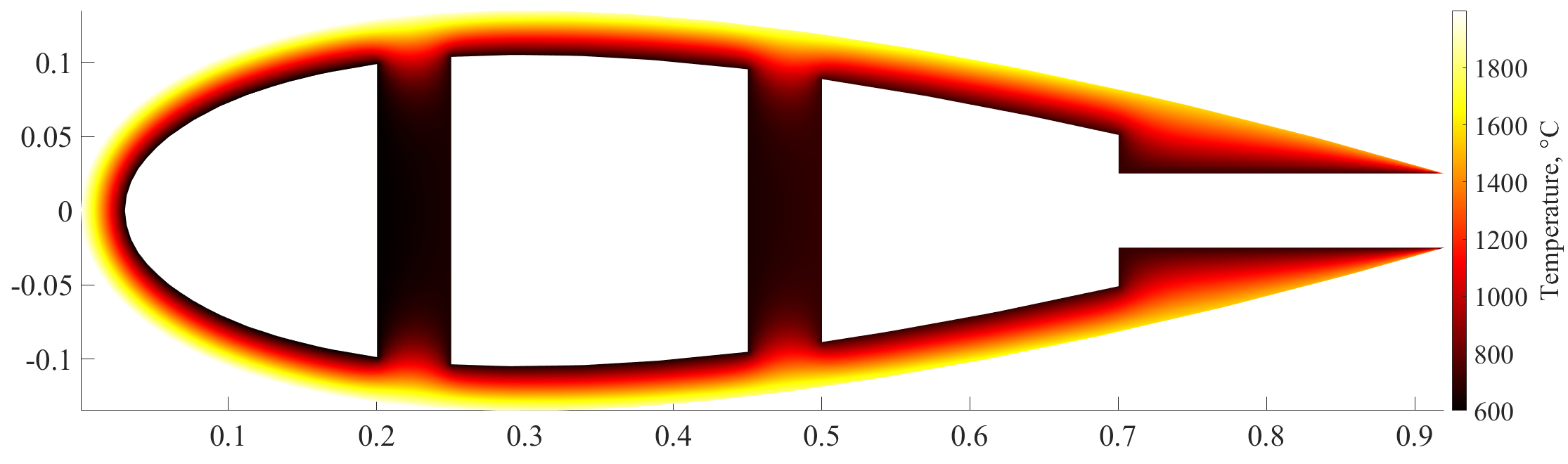}
    \caption{Coarse two-dimensional turbine blade mesh (top) and representative steady-state temperature distribution (bottom, in \(^{\circ}\)C). Numbered labels correspond to the boundaries on which Dirichlet boundary conditions are applied.}
    \label{fig:turbine_fem}
\end{figure}
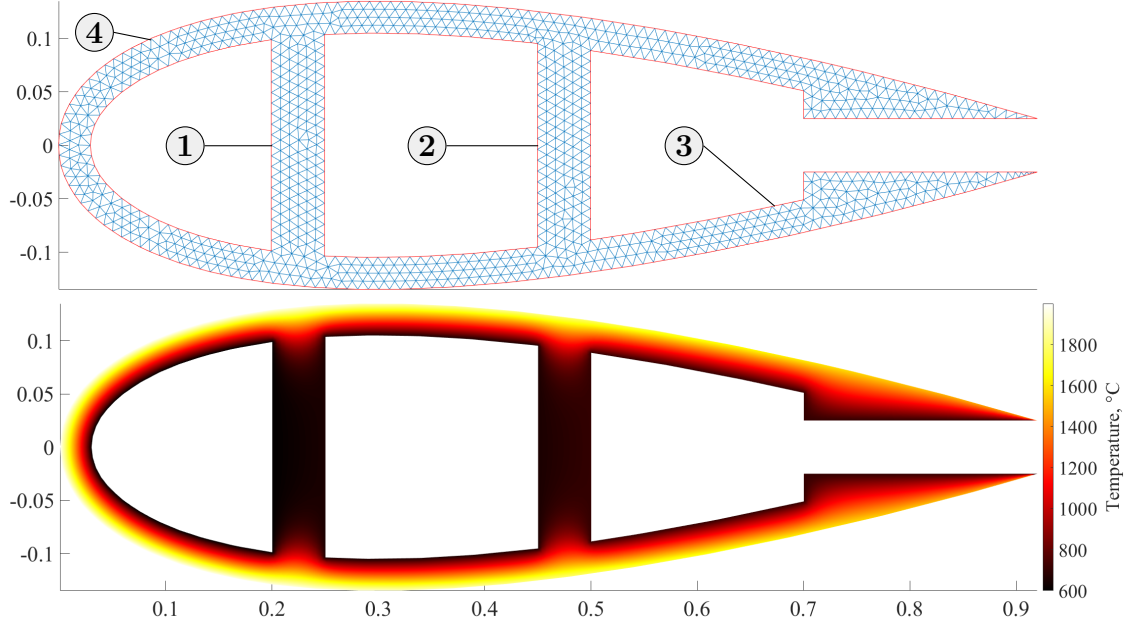
The third example is adapted from Jakeman et al.~\cite{jakeman2022adaptive}. 
It models a gas-turbine system with four disciplines: a blade heat-transfer model, a blade-lifetime model, a turbine-performance model, and an economics model. 
The multi-fidelity structure in the source work is omitted here so that the example can be used as a single-fidelity coupled-system test case. 
The model has four coupling variables and an 11-dimensional design vector \(\mathbf{x}\). The coupling variables are summarized in Table~\ref{table:turbine-coupling-vars}, and the design inputs are summarized in Tables~\ref{table:turbine-inputs} and~\ref{table:turbine-input-components}. 
The system output \(q\) is the economics coupling variable \(y_4=r_{\mathrm{econ}}\), which represents turbine revenue.

\begin{table}[h]
\centering
\begin{tabular}{ c c c c }
    \hline
    \textbf{Task Index} & \textbf{Coupling Variable} & \textbf{Bounds} & \textbf{Description} \\
    \hline\hline
    1 & $y_1=T_{bulk}$ & [1000, 1200] & Bulk Metal Temperature ($K$) \\
    2 & $y_2=t_{fail}$ & [0.01, 100] & Blade Lifetime ($hr$) \\
    3 & $y_3=P_{eng}$ & [4.8$\mathrm{e}$+6, 6.6$\mathrm{e}$+6] & Max Turbine Power ($W$) \\
    4 & $y_4=q=r_{econ}$ & [0, 1.0$\mathrm{e}$+5] & Turbine Revenue \\
    \hline
\end{tabular}
\caption{Summary of turbine model coupling variables.}
\label{table:turbine-coupling-vars}
\end{table}

\begin{table}[h]
\centering
\begin{tabular}{ c c }
    \hline
    \textbf{Design Input} & \textbf{Components} \\
    \hline\hline
    $x_1$ & $\left[ T_{c_1}, T_{c_2}, T_{c_3}, k, h_{le}, h_{te} \right]^T$ \\
    $x_2$ & $P_{lm}$ \\
    $x_3$ & $\left[ \dot{m}, T_g, F_{perf} \right]^T$ \\
    $x_4$ & $F_{econ}$ \\
    \hline
\end{tabular}
\caption{Summary of turbine model inputs.}
\label{table:turbine-inputs}
\end{table}

\begin{table}[h]
\centering
\begin{tabular}{ c c c c }
    \hline
    \textbf{Component} & \textbf{Description} & \textbf{Unit} & \textbf{Input}\\
    \hline\hline
    $T_{c_j}$   & Coolant temperature in chamber $j\in\{1,2,3\}$ & $K$      & $x_1$ \\
    $k$         & Blade thermal conductivity & $W/(m \cdot K)$                & $x_1$ \\
    $h_{le}$    & Leading edge heat transfer coefficient & $W/(m^2 \cdot K)$  & $x_1$ \\
    $h_{te}$    & Trailing edge heat transfer coefficient & $W/(m^2 \cdot K)$ & $x_1$ \\
    $P_{lm}$    & Larson-Miller parameter & -                                 & $x_2$ \\
    $\dot{m}$   & Coolant mass flow rate & $kg/s$                             & $x_3$ \\
    $T_g$       & External gas flow temperature & $K$                         & $x_3$ \\
    $F_{perf}$  & Performance factor & -                                      & $x_3$ \\
    $F_{econ}$  & Economics factor & -                                        & $x_4$ \\
    \hline
\end{tabular}
\caption{Summary of turbine model input components. Surface numbering conventions follow those in \Cref{fig:turbine_fem}.}
\label{table:turbine-input-components}
\end{table}

The coupling structure is shown in \Cref{fig:turbine_block_diagram}. 
Unlike the preceding examples, this model is purely feed-forward. 
It therefore does not require iterative fixed-point solution, but it still defines a residual system whose zero level set corresponds to the consistent model output. 
This example tests whether the residual surrogate can also recover consistent states for feed-forward systems without modifying the proposed method.

\begin{figure}[h]
    \centering
    \begin{tikzpicture}[
        node distance=2.2cm,
        block/.style={rectangle, draw, minimum height=1.2cm, minimum width=3cm, align=center},
        smallblock/.style={rectangle, draw, minimum height=0.9cm, minimum width=1.2cm, align=center},
        thickline/.style={->, thick}
    ]
        
        \node[block] (heat) {Heat Transfer\\[2pt] $y_1=f_1(x_1)$};
        \node[block, right=1.5cm of heat] (life) {Lifetime\\[2pt] $y_2=f_2(x_2,y_1)$};
        \node[block, right=1.5cm of life] (econ) {Economics\\[2pt] $y_4=f_4(x_4,y_2,y_3)$};
    
        \node[smallblock] (x1) [above=1cm of heat] {$x_1$};
        \node[smallblock] (x2) [above=1cm of life] {$x_2$};
        \node[smallblock] (x4) [above=1cm of econ] {$x_4$};
        
        \node[block, below=1cm of econ] (perf) {Performance\\[2pt] $y_3=f_3(x_3)$};
        \node[smallblock, left=1.2cm of perf] (x3) {$x_3$};
        
        \node[smallblock, right=1.5cm of econ] (q) {$q$};
        
        \draw[thickline] (x1) -- (heat);
        \draw[thickline] (x2) -- (life);
        \draw[thickline] (x4) -- (econ);
        
        \draw[thickline] (heat) -- (life);
        \draw[thickline] (life) -- (econ);
        
        \draw[thickline] (x3) -- (perf);
        \draw[thickline] (perf.north) -- (econ.south);
        
        \draw[thickline] (econ) -- (q);
    \end{tikzpicture}
    \caption{
        $\textrm N^2$ diagram of the turbine economics coupled system. 
        Inputs are summarized in Table~\ref{table:turbine-inputs}.
    }
    \label{fig:turbine_block_diagram}
\end{figure}
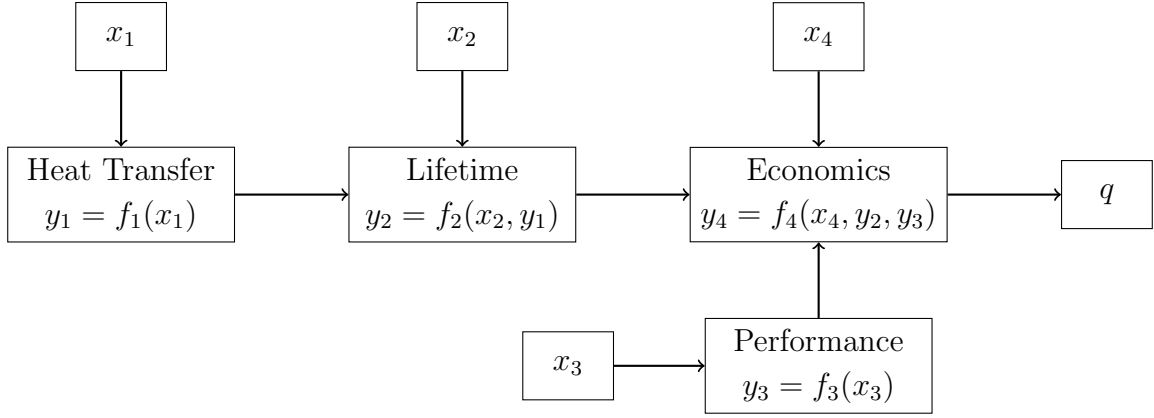

\subsubsection*{Blade heat-transfer model}
\label{sec:turbine}
The heat-transfer component models a turbine blade in a hot external flow with internal cooling channels. 
The bulk blade temperature \(T_{bulk}=y_1\) is obtained by solving the two-dimensional stationary heat equation with a finite-element method defined (with slight abuse of notation) as:
\begin{equation*}
    k\nabla^2 h(\omega)=0,
    \qquad
    \omega \in\Omega .
\end{equation*}
The boundary conditions are
\begin{gather*}
    h(\omega)=T_{c_j},
    \qquad
    \omega \in \partial\Omega_j,\quad j=1,2,3, \\
    h(\omega)
    =
    h_{te}
    +
    (h_{le}-h_{te})
    \exp
    \left(
        -4
        \frac{(1\times 10^{-3})\omega_1^2}{4 \times 10^{-6}}
    \right),
    \qquad
    \omega \in\partial\Omega_4 .
\end{gather*}
Here, $\omega_1$ is the scalar spatial component of $\omega$ in the chordwise direction.
A MATLAB function using the PDE Toolbox generates a two-dimensional blade mesh from an STL geometry file and solves the heat equation. 
The function is compiled to a Python package and called from the primary Python workflow. 
The finite element mesh and a representative temperature field are shown in \Cref{fig:turbine_fem}. 
The component output \(T_{bulk}\) is computed as the spatial average of the blade temperature field.

\subsubsection*{Blade-lifetime model}

The blade-lifetime component predicts the expected blade time to failure,
\begin{equation*}
    t_{fail}
    =
    \exp(P_{lm}/T_{bulk}-20),
\end{equation*}
where \(P_{lm}\) is the Larson--Miller parameter for creep rupture. 
Because the exponential dependence causes \(t_{fail}\) to span several orders of magnitude over the relevant range of \(T_{bulk}\), direct standardization of this response can lead to numerical precision issues. 
We therefore use the log-transformed residual
\begin{equation*}
    r_2(\mathbf{x},\{y_i\}_{i=1}^m)
    =
    \ln(t_{fail})
    -
    (P_{lm}/T_{bulk}-20).
\end{equation*}
This transformation preserves the zero level set of the original residual, so the surrogate can be trained directly on the transformed residual observations.

\subsubsection*{Turbine-performance model}

The turbine-performance component predicts the maximum engine power,
\begin{equation*}
    P_{eng}
    =
    F_{perf}
    (\dot{m}_0-N\dot{m})
    C_p T_0
    \left(
        1+T_g/T_0
        -
        2\sqrt{T_g/T_0}
    \right),
\end{equation*}
where \(F_{perf}\), \(T_g\), and \(\dot{m}\) are design inputs, and the remaining quantities are constants. 
The associated coupling variable is \(y_3=P_{eng}\).

\subsubsection*{Economics model}

The economics component predicts the turbine revenue,
\begin{equation*}
    r_{econ}
    =
    F_{econ}t_{fail}P_{eng}(c_0/1000),
\end{equation*}
where \(c_0=0.07\) is constant and \(F_{econ}\) is a design input. 
This quantity is both the fourth coupling variable and the system output,
$y_4=q=r_{econ}.$
It is not used as an input to any other component in the feed-forward turbine model.

For the feed-forward turbine results, \(n_0=8\) and \(n_{\mathrm{AL}}=92\). 
Because the model has \(m=4\) residual tasks, the total budget is 100 evaluations per residual and 400 total residual evaluations. 
The consistent state is four-dimensional and is recovered from the intersection of the four predicted zero-residual level sets.

Figure~\ref{fig:turbine_residual_matrix} visualizes a representative solution as pairwise projections of the four-dimensional residual system. 
For each pair of coupling variables, the figure compares true and surrogate-predicted zero-residual contours. 
Because this model is feed-forward, some residuals do not depend on some coupling variables. 
In the corresponding pairwise projections, the residual contours are parallel to the axis of the variable on which the residual does not depend.

\begin{figure}
    \centering
    \includegraphics[width=0.9\linewidth]{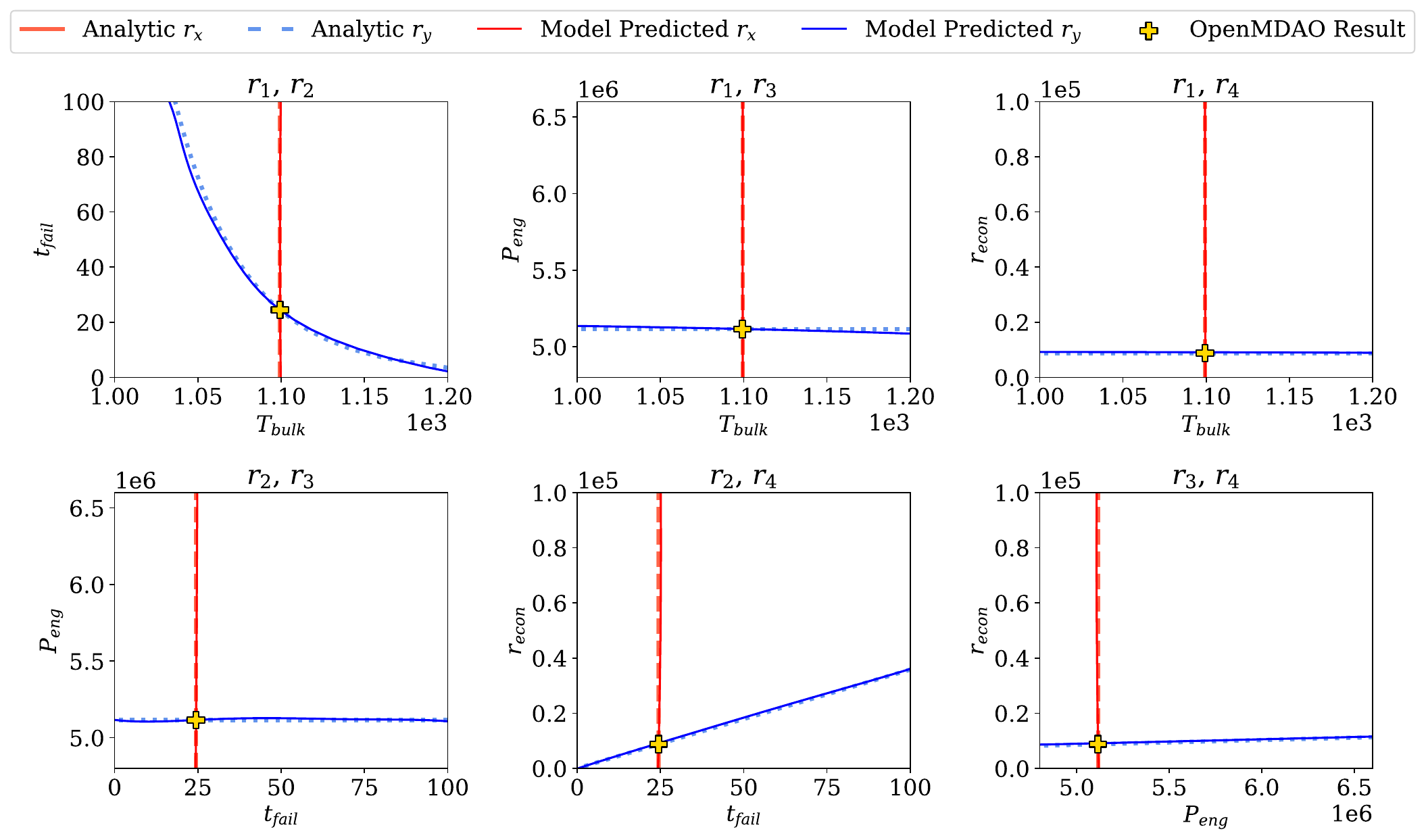}
    \caption{Pairwise projections of the turbine residual system after 92 active-learning iterations, together with the corresponding OpenMDAO-consistent state, for a representative design-variable input.}
    \label{fig:turbine_residual_matrix}
\end{figure}

Figure~\ref{fig:turbine-convergence} shows the normalized prediction-error histories for the 10 test inputs. 
After 400 total residual evaluations, entropy-based active learning achieves a mean normalized error on the order of \(10^{-3}\) and continues to decrease over the allotted budget. 
Entropy-based acquisition remains better than random acquisition on average, although the gap is smaller than in the satellite and aerostructures examples. 
A likely reason is the higher augmented dimension of this problem: the turbine surrogate is trained over 15 variables, consisting of 11 design variables and 4 coupling variables, compared with 7 variables for the satellite problem and 3 variables for the aerostructures problem. 
A larger sample budget may therefore be needed before the contour-focused behavior of the entropy acquisition function becomes dominant.

Because the feed-forward turbine system can be evaluated by a single forward pass, OpenMDAO does not require nonlinear fixed-point iteration for this case. 
Accordingly, \Cref{fig:turbine-convergence} does not include an OpenMDAO convergence history.

\begin{figure}[htb!]
    \centering
    \begin{subfigure}{\linewidth}
        \centering
        \includegraphics[width=0.99\linewidth]{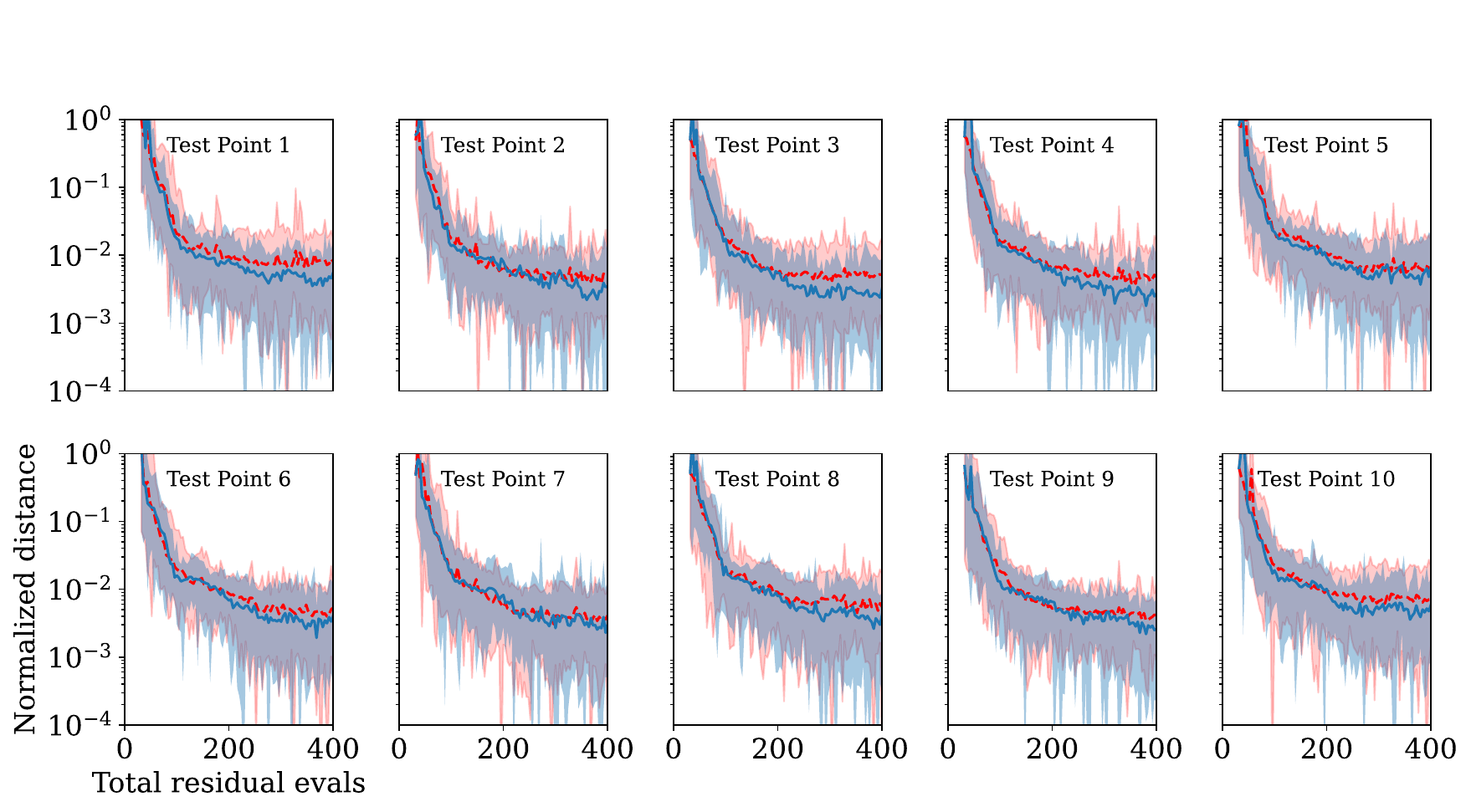}
    \end{subfigure}
    \par\medskip
    \begin{subfigure}{1\linewidth}
        \centering
        \includegraphics[width=1\linewidth]{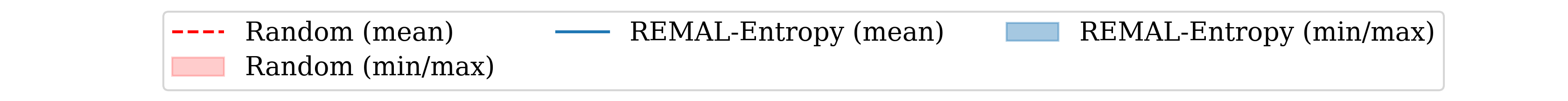}
    \end{subfigure}
    \caption{Normalized error history for the feed-forward turbine problem over 92 active-learning iterations, starting from 8 initial seed points.}
    \label{fig:turbine-convergence}
\end{figure}

\subsection{Turbine model with added feedback coupling}

The final example modifies the preceding turbine model to introduce feedback coupling. 
This produces a more general coupled system that requires iterative fixed-point solution when evaluated directly. 
Two modifications are made: first, the feed-forward coupling between the performance and economics disciplines is made bidirectional by adding a revenue-dependent term to the performance equation; second, feedback coupling between the performance and lifetime disciplines is introduced by adding performance- and lifetime-dependent terms. 
The modified lifetime and performance equations are
\begin{equation}
    \begin{gathered}
        t^\prime_{fail}
        =
        \exp
        \left(
            P_{lm}/T_{bulk}
            -
            20
            \underbrace{~+ 2(10^{-7}P_{eng})^2}_{\text{feedback}}
        \right),
        \\
        P^\prime_{eng}
        =
        F_{perf}
        (\dot{m}_0-N\dot{m})
        C_p T_0
        \left(
            1+T_g/T_0
            -
            2\sqrt{T_g/T_0}
        \right)
         \underbrace{~ +100t_{fail}^2+10^{-4}r_{econ}^2}_{\text{feedback}}.
    \end{gathered}
\end{equation}
The added feedback paths are shown in red in \Cref{fig:turbine_feedback_block_diagram}.

\begin{figure}[h]
    \centering
    \begin{tikzpicture}[
        node distance=2.2cm,
        block/.style={rectangle, draw, minimum height=1.2cm, minimum width=3cm, align=center},
        smallblock/.style={rectangle, draw, minimum height=0.9cm, minimum width=1.2cm, align=center},
        thickline/.style={->, thick}
    ]
        
        \node[block] (heat) {Heat Transfer\\[2pt] $y_1=f_1(x_1)$};
        \node[block, right=1.5cm of heat] (life) {Lifetime\\[2pt] $y_2=f_2(x_2,y_1,y_3)$};
        \node[block, right=1.5cm of life] (econ) {Economics\\[2pt] $y_4=f_4(x_4,y_2,y_3)$};
    
        \node[smallblock] (x1) [above=1cm of heat] {$x_1$};
        \node[smallblock] (x2) [above=1cm of life] {$x_2$};
        \node[smallblock] (x4) [above=1cm of econ] {$x_4$};
        
        \node[block, below=1cm of life] (perf) {Performance\\[2pt] $y_3=f_3(x_3,y_2,y_4)$};
        \node[smallblock, left=1.2cm of perf] (x3) {$x_3$};
        
        \node[smallblock, right=1.5cm of econ] (q) {$q$};
        
        \draw[thickline] (x1) -- (heat);
        \draw[thickline] (x2) -- (life);
        \draw[thickline] (x4) -- (econ);
        
        \draw[thickline] (heat) -- (life);
        \draw[thickline] (life) -- (econ);

        \draw[thickline, very thick, color=red] ([xshift=-3pt]life.south) -- ([xshift=-3pt]perf.north);
        \draw[thickline, very thick, color=red] ([xshift=3pt]perf.north) -- ([xshift=3pt]life.south);
        
        \draw[thickline] (x3) -- (perf);
        \draw[thickline] ([yshift=-3pt]perf.east) -| ([xshift=3pt]econ.south);
        \draw[thickline, very thick, color=red] ([xshift=-3pt]econ.south) |- ([yshift=3pt]perf.east);
        
        \draw[thickline] (econ) -- (q);
    \end{tikzpicture}
    \caption{Diagram of the modified turbine economics model. Added feedback couplings are shown as red arrows.}
    \label{fig:turbine_feedback_block_diagram}
\end{figure}
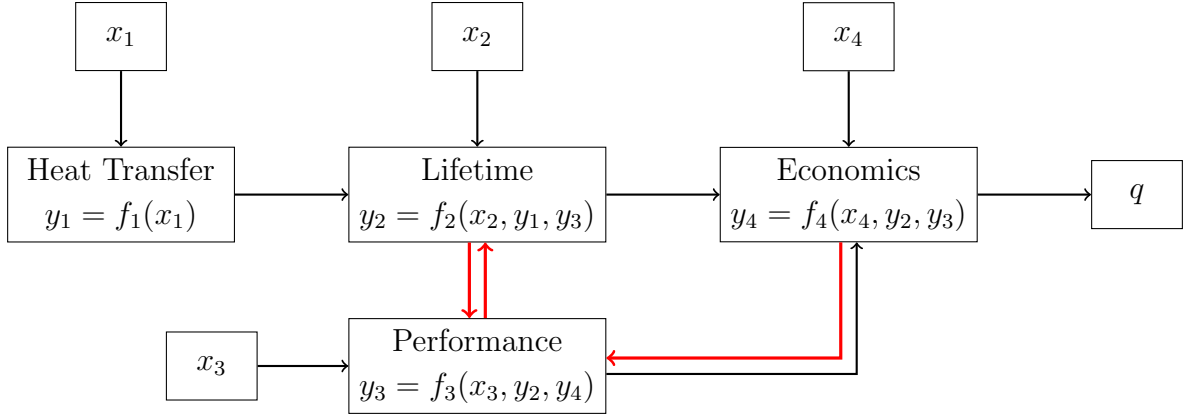

For this example, \(n_0=8\) and \(n_{\mathrm{AL}}=92\), giving 100 evaluations per residual and 400 total residual evaluations. 
Figure~\ref{fig:turbine_feedback_matrix} shows pairwise projections of the four-dimensional residual system for a representative test input. 
Compared with the feed-forward turbine model in \Cref{fig:turbine_residual_matrix}, the modified model exhibits nonlinear pairwise residual relationships for variables that are directly feedback coupled or downstream of a feedback-coupled component. 
In particular, the projections involving \((y_2,y_3)\), corresponding to \((t_{fail},P_{eng})\), and \((y_3,y_4)\), corresponding to \((P_{eng},r_{econ})\), are no longer dominated by vertical or horizontal contours.
\begin{figure}[htb!]
    \centering
    \includegraphics[width=1.0\linewidth]{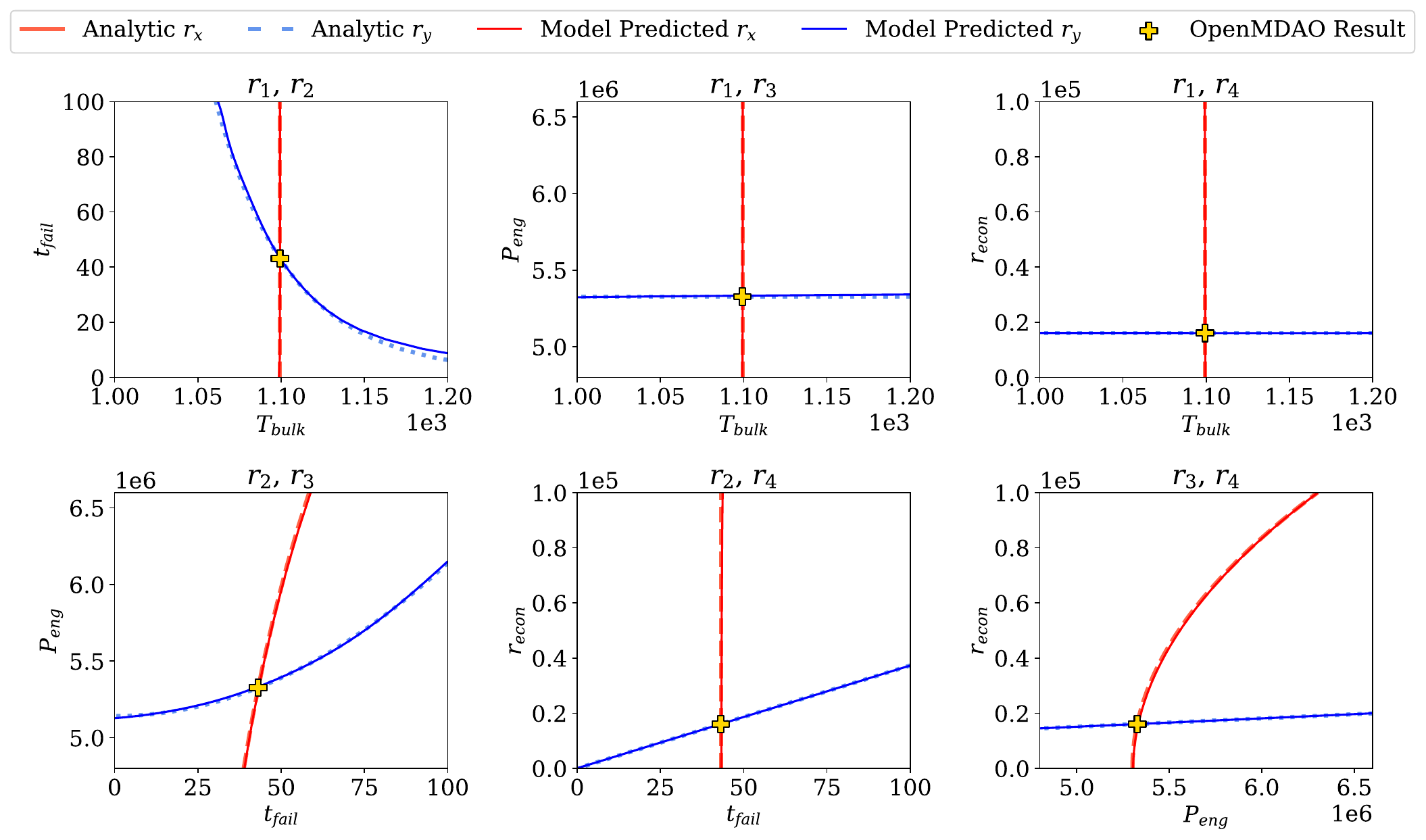}
    \caption{Pairwise projections of the residual intersections for the modified turbine system with added feedback coupling. Feedback is introduced between \(y_2\) and \(y_3\) (\(t_{fail}\) and \(P_{eng}\)) and between \(y_3\) and \(y_4\) (\(P_{eng}\) and \(r_{econ}\)).}
    \label{fig:turbine_feedback_matrix}
\end{figure}

\begin{figure}[htb!]
    \centering
    \begin{subfigure}{\linewidth}
        \centering
        \includegraphics[width=1.0\linewidth]{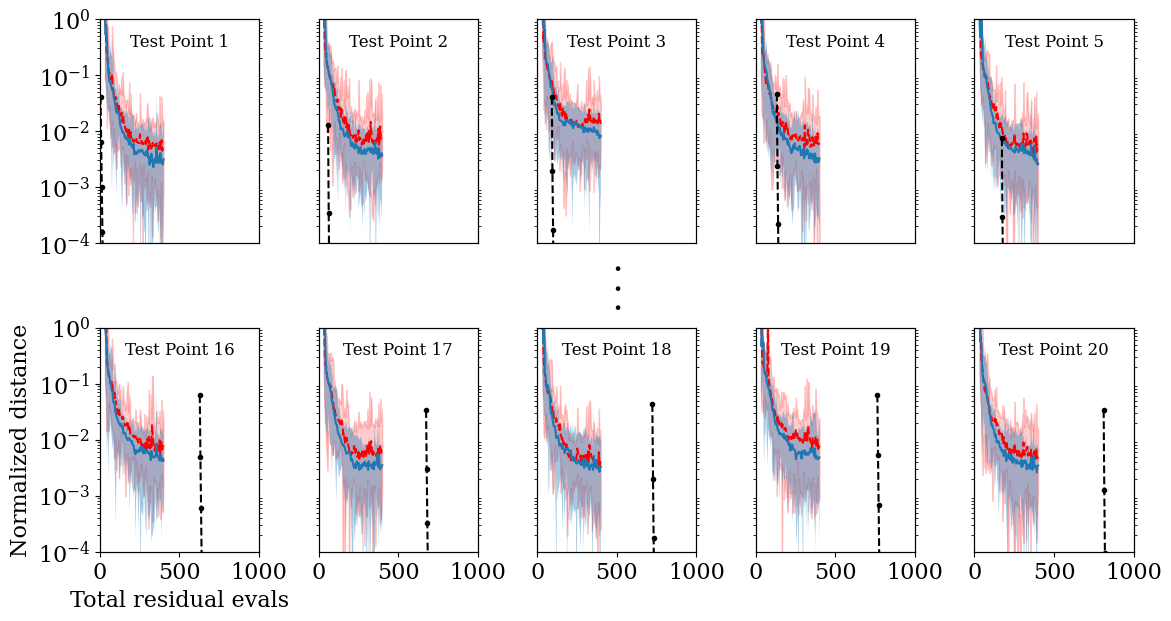}
    \end{subfigure}
    \par\medskip
    \begin{subfigure}{1\linewidth}
        \centering
        \includegraphics[width=\linewidth]{figures/legend.png}
    \end{subfigure}
    \caption{Normalized error history for the modified feedback-coupled turbine problem over 92 active-learning iterations, starting from 8 initial seed points. Results are reported over 20 repetitions and $N_{\text{test}}=20$ test inputs.}
    \label{fig:turbine-vs-openmdao}
\end{figure}

The convergence histories in \Cref{fig:turbine-vs-openmdao} show that entropy-based acquisition reduces the mean normalized fixed-point error to the order of \(10^{-3}\) after 400 total residual evaluations and outperforms random acquisition over the same budget.
Beyond this budget, the mean error using entropy-based acquisition appears to exhibit a decreasing trend.
The modified turbine problem is more difficult than the feed-forward case because the feedback terms create nonlinear dependencies among the residual contours. 
Within the 20-test-point validation budget used here, the surrogate becomes more evaluation-efficient than OpenMDAO after roughly 10 test points.
The trained surrogate can be reused for additional design points without further residual evaluations, so its relative cost continues to decrease as the number of required system analyses increases.

\section{Conclusions}
\label{section:conclusion}
This paper introduced a residual-equilibrium surrogate framework for multidisciplinary design analysis of coupled engineering systems. Rather than approximating the disciplinary solvers individually or learning a direct map from design variables to converged coupling variables, the proposed method learns the residual operator whose zero level set defines multidisciplinary consistency. This changes the surrogate modeling task from predicting a collection of converged analyses to approximating the implicit equilibrium manifold over the augmented design--coupling space.

The resulting workflow generates residual data from decoupled disciplinary evaluations, trains a multitask Gaussian process to model the residual components jointly, enriches the model with entropy-based active learning near uncertain zero-residual contours, and recovers equilibrium states for new designs through a least-residual solve on the trained surrogate. Across the satellite, aerostructural, feed-forward turbine, and feedback-coupled turbine examples, entropy-guided acquisition consistently improved over random acquisition under the same residual-evaluation budgets. The numerical results also show the intended amortized behavior of the approach: fixed-point iteration remains the more accurate tool for an isolated analysis, but the trained residual surrogate can be reused for many additional design points without further disciplinary evaluations.

These results suggest that residual-manifold modeling is a useful alternative to fixed-point iteration when many coupled-system analyses are required, as in MDO, uncertainty propagation, design-space exploration, or digital twin updating. The method also provides a common treatment of feedback-coupled and feed-forward systems, since both can be represented through residual equations whose zeros encode consistency. The stability bound further clarifies the role of residual-surrogate accuracy and root conditioning in determining the error of the recovered fixed point.

\bigskip
\noindent\textbf{Caveats.}
The main cost of the proposed approach is paid upfront. A Gaussian process residual surrogate must be trained over the augmented design--coupling space before it can be reused, and this cost is justified only when the number of downstream analyses is large enough to amortize the training effort. Within the evaluation budgets used in this paper, the surrogate does not match the accuracy of OpenMDAO fixed-point iteration, which can converge individual analyses to near machine precision and often with fewer evaluations for a single design point. The normalized errors achieved here, on the order of \(10^{-4}\) to \(10^{-3}\), are nevertheless appropriate for many mid-fidelity surrogate modeling and design-space exploration settings.

The convergence histories also indicate that the surrogate errors had not clearly plateaued within the tested budgets, so additional residual evaluations may further improve fixed-point predictions. This must be balanced against the poor scaling of exact Gaussian process inference with increasing data size~\cite{Gardner2018gpytorch, Balandat2020botorch}. Scalable GP approximations, structured kernels, or local residual surrogates are natural directions for larger coupled systems.

Finally, the present active-learning strategy acquires one new observation for each residual task at every iteration, which gives all disciplines the same evaluation budget. This is convenient but not required by the multitask formulation. Future work should allow task-dependent acquisition when disciplinary costs differ, and should incorporate multi-fidelity residual observations when cheaper approximations of selected disciplines are available. The code used for this project has been released on GitHub to facilitate future work.

\section{Acknowledgments}
This research was partially supported by a Rising Researcher award 
ICDS\_RR25\_027416
from Penn State's Institute for Computational \& Data Sciences (ICDS) (RRID:SCR\_025154). 
The authors of this work also recognize ICDS for providing access to computational research infrastructure within the Roar Core Facility (RRID: SCR\_026424).

\clearpage
\bibliography{references}

@article{renganathan2026surrogate,
  title={Surrogate-Guided Adaptive Importance Sampling for Failure Probability Estimation},
  author={Renganathan, Ashwin and Booth, Annie S},
  journal={arXiv preprint arXiv:2603.20959},
  year={2026}
}

@article{carlson2025multiobjective,
  title={Multiobjective aerodynamic design optimization of the NASA common research model},
  author={Carlson, Kade and Renganathan, Ashwin},
  journal={Aerospace Science and Technology},
  pages={111120},
  year={2025},
  publisher={Elsevier}
}

@article{jones1998efficient,
  title={Efficient global optimization of expensive black-box functions},
  author={Jones, Donald R and Schonlau, Matthias and Welch, William J},
  journal={Journal of Global optimization},
  volume={13},
  number={4},
  pages={455--492},
  year={1998},
  publisher={Springer}
}

@article{frazier2008knowledge,
  title={A knowledge-gradient policy for sequential information collection},
  author={Frazier, Peter I and Powell, Warren B and Dayanik, Savas},
  journal={SIAM Journal on Control and Optimization},
  volume={47},
  number={5},
  pages={2410--2439},
  year={2008},
  publisher={SIAM}
}

@article{renganathan2021lookahead,
  title={Lookahead acquisition functions for finite-horizon time-dependent bayesian optimization and application to quantum optimal control},
  author={Renganathan, S Ashwin and Larson, Jeffrey and Wild, Stefan M},
  journal={arXiv preprint arXiv:2105.09824},
  year={2021}
}

@inproceedings{renganathan2025qpots,
  title={qPOTS: Efficient Batch Multiobjective Bayesian Optimization via Pareto Optimal Thompson Sampling},
  author={Renganathan, Ashwin and Carlson, Kade},
  booktitle={International Conference on Artificial Intelligence and Statistics},
  pages={4051--4059},
  year={2025},
  organization={PMLR}
}

@inproceedings{gemseo,
	title        = {GEMS: a Python library for automation of multidisciplinary design optimization process generation},
	author       = {Gallard, Fran{\c{c}}ois and Vanaret, Charlie and Gu{\'e}not, Damien and Gachelin, Vincent and Lafage, R{\'e}mi and Pauwels, Benoit and Barjhoux, Pierre-Jean and Gazaix, Anne},
	year         = 2018,
	booktitle    = {2018 AIAA/ASCE/AHS/ASC Structures, Structural Dynamics, and Materials Conference},
	pages        = {0657}
}

@inproceedings{swersky2013multitask,
	title        = {Multi-Task Bayesian Optimization},
	author       = {Swersky, Kevin and Snoek, Jasper and Adams, Ryan P},
	year         = 2013,
	booktitle    = {Advances in Neural Information Processing Systems},
	publisher    = {Curran Associates, Inc.},
	volume       = 26,
	editor       = {C.J. Burges and L. Bottou and M. Welling and Z. Ghahramani and K.Q. Weinberger}
}

@article{chevalier2014uncertainty,
	title        = {Fast Parallel Kriging-Based Stepwise Uncertainty Reduction With Application to the Identification of an Excursion Set},
	author       = {Clément Chevalier and Julien Bect and David Ginsbourger and Emmanuel Vazquez and Victor Picheny and Yann Richet},
	year         = 2014,
	journal      = {Technometrics},
	publisher    = {Taylor \& Francis},
	volume       = 56,
	number       = 4,
	pages        = {455--465}
}

@article{cole2023entropycontour,
	title        = {Entropy-based adaptive design for contour finding and estimating reliability},
	author       = {D. Austin Cole and Robert B. Gramacy and James E. Warner and Geoffrey F. Bomarito and Patrick E. Leser and William P. Leser},
	year         = 2023,
	journal      = {Journal of Quality Technology},
	publisher    = {Taylor \& Francis},
	volume       = 55,
	number       = 1,
	pages        = {43--60}
}

@article{ranjan2008contour,
	title        = {Sequential Experiment Design for Contour Estimation From Complex Computer Codes},
	author       = {Pritam Ranjan and Derek Bingham and George Michailidis},
	year         = 2008,
	journal      = {Technometrics},
	publisher    = {Taylor \& Francis},
	volume       = 50,
	number       = 4,
	pages        = {527--541}
}

@inproceedings{marques2018contourentropy,
	title        = {Contour location via entropy reduction leveraging multiple information sources},
	author       = {Marques, Alexandre and Lam, Remi and Willcox, Karen},
	year         = 2018,
	booktitle    = {Advances in Neural Information Processing Systems},
	publisher    = {Curran Associates, Inc.},
	volume       = 31,
	editor       = {S. Bengio and H. Wallach and H. Larochelle and K. Grauman and N. Cesa-Bianchi and R. Garnett}
}

@article{booth2025contour,
	title        = {Contour location for reliability in airfoil siulation experiments using deep Gaussian processes},
	author       = {Booth, Annie S. and Renganathan, S. Ashwin and Gramacy, Robert B.},
	year         = 2025,
	journal      = {The Annals of Applied Statistics},
	volume       = 19,
	pages        = {191--211}
}

@article{booth2025failure,
	title        = {Two-stage design for failure probability estimation with Gaussian process surrogates},
	author       = {Booth, Annie and Renganathan, Ashwin},
	year         = 2025,
	month        = 10,
	journal      = {Journal of Quality Technology},
	volume       = 57,
	pages        = {1--17}
}

@article{kyzyurova2018gpchain,
	title        = {Coupling Computer Models through Linking Their Statistical Emulators},
	author       = {Kyzyurova, Ksenia N. and Berger, James O. and Wolpert, Robert L.},
	year         = 2018,
	journal      = {SIAM/ASA Journal on Uncertainty Quantification},
	volume       = 6,
	number       = 3,
	pages        = {1151--1171}
}

@article{sanson2019gpchain,
	title        = {Systems of Gaussian process models for directed chains of solvers},
	author       = {Francois Sanson and Olivier {Le Maitre} and Pietro Marco Congedo},
	year         = 2019,
	journal      = {Computer Methods in Applied Mechanics and Engineering},
	volume       = 352,
	pages        = {32--55},
	issn         = {0045-7825}
}

@article{jakeman2022adaptive,
	title        = {Adaptive Experimental Design for Multi-Fidelity Surrogate Modeling of Multi-Disciplinary Systems},
	author       = {Jakeman, John D. and Friedman, Sam and Eldred, Michael and Tamellini, Lorenzo and Gorodetsky, Alex A. and Allaire, Douglas},
	year         = 2022,
	journal      = {International Journal for Numerical Methods in Engineering},
	volume       = 123,
	number       = 12,
	pages        = {2760--2790}
}

@article{mckay1979lhs,
	title        = {A Comparison of Three Methods for Selecting Values of Input Variables in the Analysis of Output from a Computer Code},
	author       = {M. D. McKay and R. J. Beckman and W. J. Conover},
	year         = 1979,
	journal      = {Technometrics},
	publisher    = {[Taylor & Francis, Ltd., American Statistical Association, American Society for Quality]},
	volume       = 21,
	number       = 2,
	pages        = {239--245},
	issn         = {00401706},
	urldate      = {2026-01-27}
}

@book{gramacy2020surrogates,
	title        = {Surrogates: Gaussian process modeling, design and optimization for the applied sciences},
	author       = {Gramacy, Robert B.},
	year         = 2020,
	publisher    = {CRC Press}
}

@book{rasmussen,
	title        = {Gaussian Processes for Machine Learning},
	author       = {Rasmussen, C.E. and Williams, K.I.},
	year         = 2006,
	publisher    = {the MIT Press}
}

@book{mdobook,
	title        = {Engineering Design Optimization},
	author       = {Martins, Joaquim R.R.A. and Ning, Andrew},
	year         = 2021,
	publisher    = {Cambridge University Press}
}

@inproceedings{ghoreishi2020aerostructures,
	title        = {Bayesian Optimization for Efficient Design of Uncertain Coupled Multidisciplinary Systems},
	author       = {Ghoreishi, Seyede Fatemeh and Imani, Mahdi},
	year         = 2020,
	booktitle    = {2020 American Control Conference (ACC)},
	pages        = {3412--3418}
}

@article{Hennig2011,
	title        = {Entropy Search for Information-Efficient Global Optimization},
	author       = {Hennig, Philipp and Schuler, Christian},
	year         = 2011,
	month        = 12,
	journal      = {Journal of Machine Learning Research},
	volume       = 13
}

@article{Gray2019,
	title        = {OpenMDAO: an open-source framework for multidisciplinary design, analysis, and optimization},
	author       = {Gray, Justin S. and Hwang, John T. and Martins, Joaquim R. R. A. and Moore, Kenneth T. and Naylor, Bret A.},
	year         = 2019,
	month        = {Apr},
	day          = {01},
	journal      = {Structural and Multidisciplinary Optimization},
	volume       = 59,
	number       = 4,
	pages        = {1075--1104},
	issn         = {1615-1488}
}

@inproceedings{Balandat2020botorch,
	title        = {{BoTorch: A Framework for Efficient Monte-Carlo Bayesian Optimization}},
	author       = {Balandat, Maximilian and Karrer, Brian and Jiang, Daniel R. and Daulton, Samuel and Letham, Benjamin and Wilson, Andrew Gordon and Bakshy, Eytan},
	year         = 2020,
	booktitle    = {Advances in Neural Information Processing Systems 33}
}

@inproceedings{Gardner2018gpytorch,
	title        = {GPyTorch: Blackbox Matrix-Matrix Gaussian Process Inference with GPU Acceleration},
	author       = {Gardner, Jacob and Pleiss, Geoff and Weinberger, Kilian Q and Bindel, David and Wilson, Andrew G},
	year         = 2018,
	booktitle    = {Advances in Neural Information Processing Systems},
	publisher    = {Curran Associates, Inc.},
	volume       = 31,
	editor       = {S. Bengio and H. Wallach and H. Larochelle and K. Grauman and N. Cesa-Bianchi and R. Garnett}
}

@inproceedings{Bonilla2007multitask,
	title        = {Multi-task Gaussian Process Prediction},
	author       = {Bonilla, Edwin V and Chai, Kian and Williams, Christopher},
	year         = 2007,
	booktitle    = {Advances in Neural Information Processing Systems},
	publisher    = {Curran Associates, Inc.},
	volume       = 20,
	editor       = {J. Platt and D. Koller and Y. Singer and S. Roweis}
}

@article{Sankararaman2012satellite,
	title        = {Likelihood-Based Approach to Multidisciplinary Analysis Under Uncertainty},
	author       = {Sankararaman, Shankar and Mahadevan, Sankaran},
	year         = 2012,
	month        = {03},
	journal      = {Journal of Mechanical Design},
	volume       = 134,
	number       = 3,
	pages        = {031008},
	issn         = {1050-0472}
}

@article{Yao2012MDO,
	title        = {A surrogate based multistage-multilevel optimization procedure for multidisciplinary design optimization},
	author       = {Yao, Wen and Chen, Xiaoqian and Ouyang, Qi and van Tooren, Michel},
	year         = 2012,
	month        = {Apr},
	day          = {01},
	journal      = {Structural and Multidisciplinary Optimization},
	volume       = 45,
	number       = 4,
	pages        = {559--574},
	issn         = {1615-1488}
}

@article{Hu2017mdra,
	title        = {A surrogate modeling approach for reliability analysis of a multidisciplinary system with spatio-temporal output},
	author       = {Hu, Zhen and Mahadevan, Sankaran},
	year         = 2017,
	month        = {Sep},
	day          = {01},
	journal      = {Structural and Multidisciplinary Optimization},
	volume       = 56,
	number       = 3,
	pages        = {553--569},
	issn         = {1615-1488}
}

@article{simpson2001kriging,
	title        = {Kriging Models for Global Approximation in Simulation-Based Multidisciplinary Design Optimization},
	author       = {Simpson, Timothy W. and Mauery, Timothy M. and Korte, John J. and Mistree, Farrokh},
	year         = 2001,
	journal      = {AIAA Journal},
	volume       = 39,
	number       = 12,
	pages        = {2233--2241}
}

@book{forrester2008engineering,
	title        = {Engineering Design via Surrogate Modelling: A Practical Guide},
	author       = {Forrester, Alexander I. J. and S{\'o}bester, Andr{\'a}s and Keane, Andy J.},
	year         = 2008,
	publisher    = {Wiley},
	address      = {Hoboken, NJ}
}

@article{martins2013multidisciplinary,
	title        = {Multidisciplinary Design Optimization: A Survey of Architectures},
	author       = {Martins, Joaquim R. R. A. and Lambe, Andrew B.},
	year         = 2013,
	journal      = {AIAA Journal},
	volume       = 51,
	number       = 9,
	pages        = {2049--2075}
}

@article{dubreuil2020towards,
	title        = {Towards an Efficient Global Multidisciplinary Design Optimization Algorithm},
	author       = {Dubreuil, Sylvain and Bartoli, Nathalie and Gogu, Christian and Lefebvre, Thierry},
	year         = 2020,
	journal      = {Structural and Multidisciplinary Optimization},
	volume       = 62,
	number       = 4,
	pages        = {1739--1765}
}

@article{cardoso2024constrained,
	title        = {Constrained Efficient Global Multidisciplinary Design Optimization Using Adaptive Disciplinary Surrogate Enrichment},
	author       = {Cardoso, In{\^e}s and Dubreuil, Sylvain and Bartoli, Nathalie and Gogu, Christian and Sala{\"u}n, Michel},
	year         = 2024,
	journal      = {Structural and Multidisciplinary Optimization},
	volume       = 67,
	number       = 2,
	pages        = 23
}

@article{berthelin2022disciplinary,
	title        = {Disciplinary Proper Orthogonal Decomposition and Interpolation for the Resolution of Parameterized Multidisciplinary Analysis},
	author       = {Berthelin, Gaspard and Dubreuil, Sylvain and Sala{\"u}n, Michel and Bartoli, Nathalie and Gogu, Christian},
	year         = 2022,
	journal      = {International Journal for Numerical Methods in Engineering},
	volume       = 123,
	number       = 15,
	pages        = {3594--3626}
}

@article{chaudhuri2018multifidelity,
	title        = {Multifidelity Uncertainty Propagation via Adaptive Surrogates in Coupled Multidisciplinary Systems},
	author       = {Chaudhuri, Anirban and Lam, Remi and Willcox, Karen},
	year         = 2018,
	journal      = {AIAA Journal},
	volume       = 56,
	number       = 1,
	pages        = {235--249}
}

@article{ghoreishi2017adaptive,
	title        = {Adaptive Uncertainty Propagation for Coupled Multidisciplinary Systems},
	author       = {Ghoreishi, Seyede Fatemeh and Allaire, Douglas L.},
	year         = 2017,
	journal      = {AIAA Journal},
	volume       = 55,
	number       = 11,
	pages        = {3940--3950}
}

@article{ghoreishi2021bayesian,
	title        = {Bayesian Surrogate Learning for Uncertainty Analysis of Coupled Multidisciplinary Systems},
	author       = {Ghoreishi, Seyede Fatemeh and Imani, Mahdi},
	year         = 2021,
	journal      = {Journal of Computing and Information Science in Engineering},
	volume       = 21,
	number       = 4,
	pages        = {041009}
}

@article{asadi2024active,
	title        = {Active Learning for Efficient Data Acquiring in Coupled Multidisciplinary Systems},
	author       = {Asadi, Negar and Ghoreishi, Seyede Fatemeh},
	year         = 2024,
	journal      = {IFAC-PapersOnLine},
	volume       = 58,
	number       = 28,
	pages        = {114--119}
}

@article{baptista2018optimal,
	title        = {Optimal Approximations of Coupling in Multidisciplinary Models},
	author       = {Baptista, Ricardo and Marzouk, Youssef and Willcox, Karen and Peherstorfer, Benjamin},
	year         = 2018,
	journal      = {AIAA Journal},
	volume       = 56,
	number       = 6,
	pages        = {2412--2428}
}

@article{alvarez2012kernels,
	title        = {Kernels for Vector-Valued Functions: A Review},
	author       = {{\'A}lvarez, Mauricio A. and Rosasco, Lorenzo and Lawrence, Neil D.},
	year         = 2012,
	journal      = {Foundations and Trends in Machine Learning},
	volume       = 4,
	number       = 3,
	pages        = {195--266}
}

@article{renganathan2023camera,
	title        = {CAMERA: A method for cost-aware, adaptive, multifidelity, efficient reliability analysis},
	author       = {Renganathan, S Ashwin and Rao, Vishwas and Navon, Ionel M},
	year         = 2023,
	journal      = {Journal of Computational Physics},
	publisher    = {Elsevier},
	volume       = 472,
	pages        = 111698
}

@inproceedings{booth2024actively,
	title        = {Actively learning deep Gaussian process models for failure contour and probability estimation.},
	author       = {Booth, Annie S and Gramacy, Robert and Renganathan, Ashwin},
	year         = 2024,
	booktitle    = {AIAA SCITECH 2024 Forum},
	pages        = {0577}
}

@inproceedings{renganathan2024efficient,
	title        = {Efficient reliability analysis with multifidelity Gaussian processes and normalizing flows},
	author       = {Renganathan, Ashwin},
	year         = 2024,
	booktitle    = {AIAA SCITECH 2024 Forum},
	pages        = {0576}
}

@book{nocedal2006numerical,
	title        = {Numerical optimization},
	author       = {Nocedal, Jorge and Wright, Stephen J},
	year         = 2006,
	publisher    = {Springer}
}

@article{sobieszczanski-sobieski_sensitivity_1990,
    title = {Sensitivity of complex, internally coupled systems},
    volume = {28},
    issn = {0001-1452, 1533-385X},
    url = {https://arc.aiaa.org/doi/10.2514/3.10366},
    doi = {10.2514/3.10366},
    language = {en},
    number = {1},
    urldate = {2023-03-08},
    journal = {AIAA Journal},
    author = {Sobieszczanski-Sobieski, Jaroslaw},
    month = jan,
    year = {1990},
    pages = {153--160},
}

\end{document}